\documentclass[%
 reprint,
 amsmath,amssymb,
 aps,
pre,
]{revtex4-1}

\usepackage{graphicx}
\usepackage{dcolumn}
\usepackage{bm}

\usepackage{ntheorem}
\newtheorem{theorem}{Theorem}

\newtheorem{lemma}[theorem]{Lemma}

\newtheorem{remark}[theorem]{Remark}

\newenvironment{proof}[1][Proof]{\textbf{#1.} }{\ \rule{0.5em}{0.5em}}

\usepackage{amsmath}%
\usepackage{amsfonts}%
\usepackage{amssymb}%
\usepackage[caption=false]{subfig}
\usepackage{bbm}
\usepackage[bookmarksnumbered=true]{hyperref}

\newcommand{\Prob}{\mathbb{P}}

\begin{document}

\preprint{APS/123-QED}

\title{The impact of a network split on cascading failure processes}

\author{Fiona Sloothaak}
\author{Sem C. Borst}
 \altaffiliation[Also at ]{Nokia Bell Labs, Murray Hill, New Jersey, USA.}
\author{Bert Zwart}%
 \altaffiliation[Also at ]{Centrum Wiskunde en Informatica (CWI), Amsterdam, The Netherlands.}
\affiliation{Eindhoven University of Technology, Eindhoven, The Netherlands.}%

\date{\today}

\begin{abstract}
Cascading failure models are typically used to capture the phenomenon where failures possibly trigger further failures in succession, causing knock-on effects. In many networks this ultimately leads to a disintegrated network where the failure propagation continues independently across the various components. In order to gain insight in the impact of network splitting on cascading failure processes, we extend a well-established cascading failure model for which the number of failures obeys a power-law distribution. We assume that a single line failure immediately splits the network in two components, and  examine its effect on the power-law exponent. The results provide valuable qualitative insights that are crucial first steps towards understanding more complex network splitting scenarios.
\end{abstract}

\maketitle


\section{Introduction}
Cascading failure models are used to describe systems of interconnected components where initial failures can trigger subsequent failures of other components. Despite the deceptively simple appearance, these models capture an extraordinary richness of different behaviors and have proven to be crucial in many application areas, such as material science, computer networks, traffic networks, earthquake dynamics, and power transmission systems~\cite{Pradhan2010}. It is therefore not surprising that cascading failures have received a lot of attention throughout the years~\cite{Albert2000,Albert2002,Motter2002,Watts2002,Crucitti2004,Motter2004,Heide2008, Mirzasoleiman2011,Witthaut2015}

In particular, the application domain of power systems has received increasing attention over the last fifteen years~\cite{Albert2004,Kinney2005,Bienstock2015,Rohden2016}, and also provides the main inspiration for our research. Power grids have grown significantly in size and complexity. Moreover, various recent advances, such as the rise of renewable sources, have considerably increased the volatility in these systems. The occurrences of severe blackouts have increased rapidly around the world, in a time where society relies on a reliable power grid more and more. Notorious examples include the Northeast Blackout of 2003, the India Blackout of 2012 and the South Australia Blackouts of 2016 and 2017. Such catastrophic events cause significant economic and social disruption, and the analysis of severe blackouts has therefore become a crucial part of transmission grid planning and operations~\cite{Newman2011,Wang2015}.

Blackouts often occur through a cascade of failures that accelerate and outstrip control capabilities~\cite{Bienstock2015}. The failure mechanism causing a power outage entails long and complex sequences of failures, making the analysis of the failure propagation extremely difficult. Simulation techniques are typical approaches in order to obtain a better understanding of the cascading failure process. However, standard Monte-Carlo simulation may become computationally intractable due to the low probability of a blackout event and the huge size of the network (curse of dimensionality)~\cite{Guo2016}. Nevertheless, rare-event simulation, such as importance sampling and splitting~\cite{Kim2013,Shortle2013}, can be used to overcome these issues and analyze fairly complex cascading failure models. 

Another approach, often used in practice, involves extensive scenario analysis where one tests whether the failure of one line, or a combination of several lines, is likely to result in a major blackout~\cite{Bienstock2015,Soltan2014}. There are two typical ways to measure the reliability: worst-case analysis where the initial contingency is targeted to be the most vulnerable component, or one where it is chosen uniformly at random in order to understand the typical behavior on average~\cite{Bienstock2015}. The fact that this distinction is critical has also been illustrated for complex networks: they are rather resilient to random attacks, while relatively vulnerable to targeted attacks~\cite{Cohen2000,Cohen2001}.

Although advanced simulation techniques and scenario testing approaches have proven indispensable, they provide little physical insight in the mechanism leading to a severe blackout. In contrast, macroscopic models, such as~\cite{Motter2002,Watts2002,Crucitti2004,Dobson2004,BlanchetUnpublished}, focus on a few essential characteristics to obtain more qualitative insights. Such insights help in gaining a deeper understanding of the failure propagation. In particular, Dobson et al.~\cite{Dobson2004,Dobson2005} construct a simple cascading failure model that captures four salient features of large blackouts: the large number of components, the initial disturbance stressing the network, the component failure when its capacity is exceeded and the additional loading of other components when a component failure has occurred. This results in a tractable model that allows for a rigorous derivation of the number of component failures.

In this paper, we extend these models, allowing for another distinctive feature observed in occurrences of large blackouts: network splitting. Successive line failures may cause the network to disintegrate in disjoint components. Once a network split has occurred, the failure propagation continues independently among the various components. Network splitting is also known as islanding, and is sometimes used as a tool in power systems to prevent blackouts to cascade to large-scale proportions~\cite{Bienstock2015}. Our results show the impact of islanding on the power-law exponent.

Specifically, we consider a network consisting of two star components connected by a single line, see Figure~\ref{fig:ModelFigure}. Each line has an initial load that is exceeded by the capacity by a random margin. The cascading failure process is triggered by the failure of the line bridging the two components, causing all lines to be additionally loaded. When this load surge causes the capacity to be exceeded on a line, it fails. Every consecutive line failure causes all surviving lines connected to it to receive another supplementary load increase, possibly triggering massive knock-on effects. We emphasize that due to the network structure, no network splitting occurs after the failure of the bridging line. Therefore, the cascading failure propagation continues independently in the two components until the capacities at the  surviving lines in both components are sufficient to meet the load surges. A detailed description of the model is given in Section~\ref{sec:ModelDescription}.

We measure the reliability of the network by the probability that the total number of line failures exceeds a certain threshold, which we refer to as the \emph{exceedance probability}. This objective is well-understood in the case of a single star network under certain assumptions~\cite{Sloothaak2016}. That is, there is an initial disruption causing all lines to be additionally loaded and every consecutive line failure causes subsequent load surges to all surviving lines. Under a particular condition, the exceedance probability obeys a power-law distribution with exponent $-1/2$. This type of behavior is of strong interest as it appears in empirical analyses of historic blackout data~\cite{Carreras2000,Carreras2004,Talukdar2003,Newman2011}. This heavy-tailed behavior reflects a relatively high risk of having severe blackouts.

The objective of this paper is to examine the impact of including a single immediate network split on the exceedance probability. It turns out that the power-law property, which appeared in case of a star topology, mostly prevails. However, the splitting feature may possibly change the prefactor and the exponent depending on the threshold and component sizes. The results can intuitively be interpreted as follows. When the threshold is sufficiently smaller than the size of the smaller component, the threshold is most likely been exceeded in just one of the components alone. If the threshold is approximately between the size of the smaller and the larger component, the threshold is most likely exceeded in the bigger component alone. In both cases, this property will imply that the power-law exponent is $-1/2$ as is the case of a single star network. For larger threshold values, both components need to have a significant number of line failures. Consequently, it is much less likely for the threshold to be exceeded, which causes a phase-transition: the power-law exponent is reduced to~$-1$. This provides a possible explanation of why also other power-law exponents appear in empirical data analyses~\cite{Hines2009}.

Our methodology uses an asymptotic analysis for the sum of two independent quasi-binomially distributed random variables. We distinguish between different cases: the balanced case where the sizes of both components are of the same order of magnitude, and the disparate case where one is of a smaller order. Preliminary results appeared in a conference paper~\cite{Sloothaak2017}, without proofs, where we focused on an approximation scheme for the exceedance probability and compared this to simulation results. In the present paper, we rigorously prove the asymptotic behavior. In the analysis many subtleties need to be accounted for, which are most apparent when the threshold is close to the size of the larger component. Our analysis in Section~\ref{sec:Main results} aims to provide physical insights in these subtleties.

The paper is organized as follows. In Section~\ref{sec:ModelDescription} we describe our model in more detail and review some known results for the single star network. We state our main results in Section~\ref{sec:Main results} and provide a high-level interpretation. The proofs of the main results are covered in Section~\ref{sec:ProofsOfMainResults}. We conclude this paper with a discussion of some future research directions in Section~\ref{sec:Future}.

\section{Model description and preliminaries}
\label{sec:ModelDescription}
We consider a network with $n+2$ nodes, where $n$ is large. The network consists of two components connected by a single line. The smaller component  consists of $l:=l_n$ lines, whereas the other component has $n-l$ lines, and hence~$l \leq n-l$. Each line has a limited capacity for the amount of load it can carry before it fails. We assume that the network is initially stable in the sense that every line has enough capacity to carry its load. The difference between the initial load and capacity is called the \emph{surplus capacity}, and we assume it to be independent and standard uniformly distributed at each of the $n$ lines. A visual representation of the model is given in Figure~\ref{fig:ModelFigure}

\begin{figure}[htb]
\centering
\includegraphics[width=8cm]{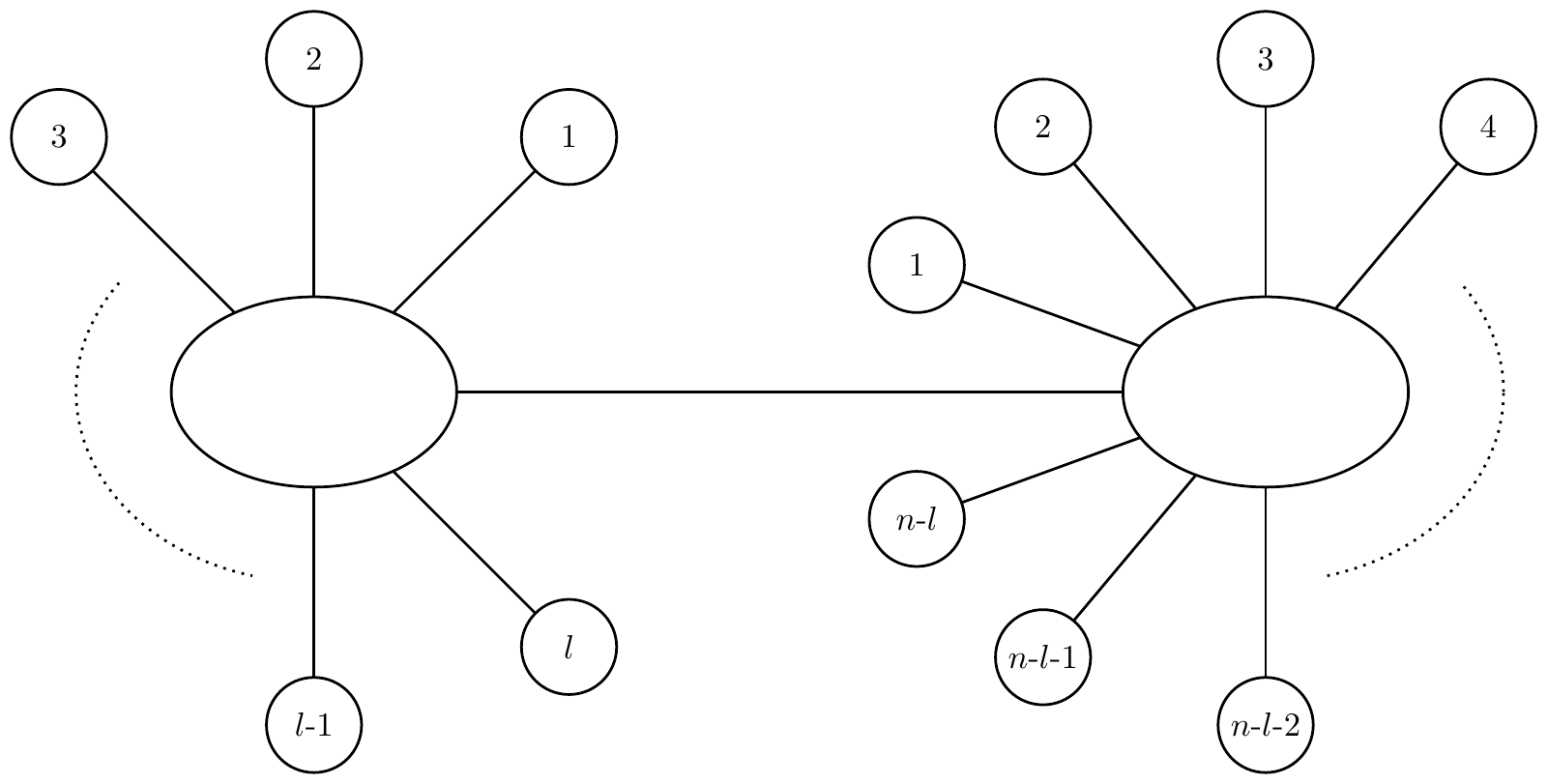}
\caption{Visual representation of the network.}
\label{fig:ModelFigure}
\end{figure}

\begin{figure*}
\centering
\subfloat[Case $0 \leq \alpha<\beta$.\label{fig:SumCaseA}]{%
  \includegraphics[width=.3\linewidth]{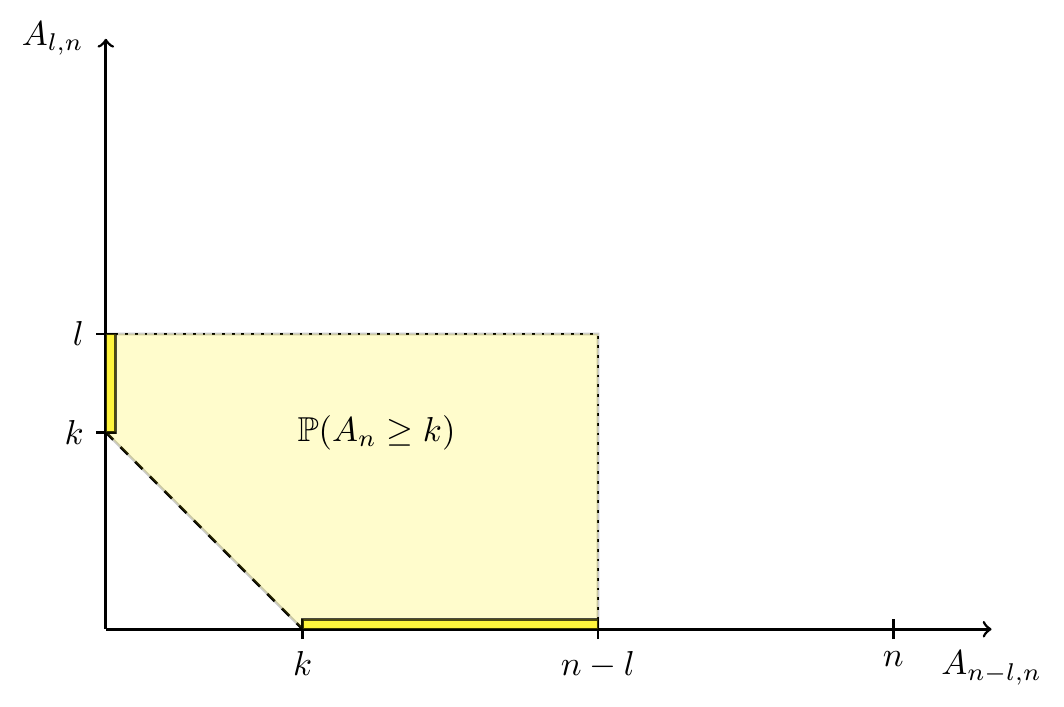}%
	}\hfill
\subfloat[Case $\beta \leq \alpha<1-\beta$.\label{fig:SumCaseB}]{%
  \includegraphics[width=.3\linewidth]{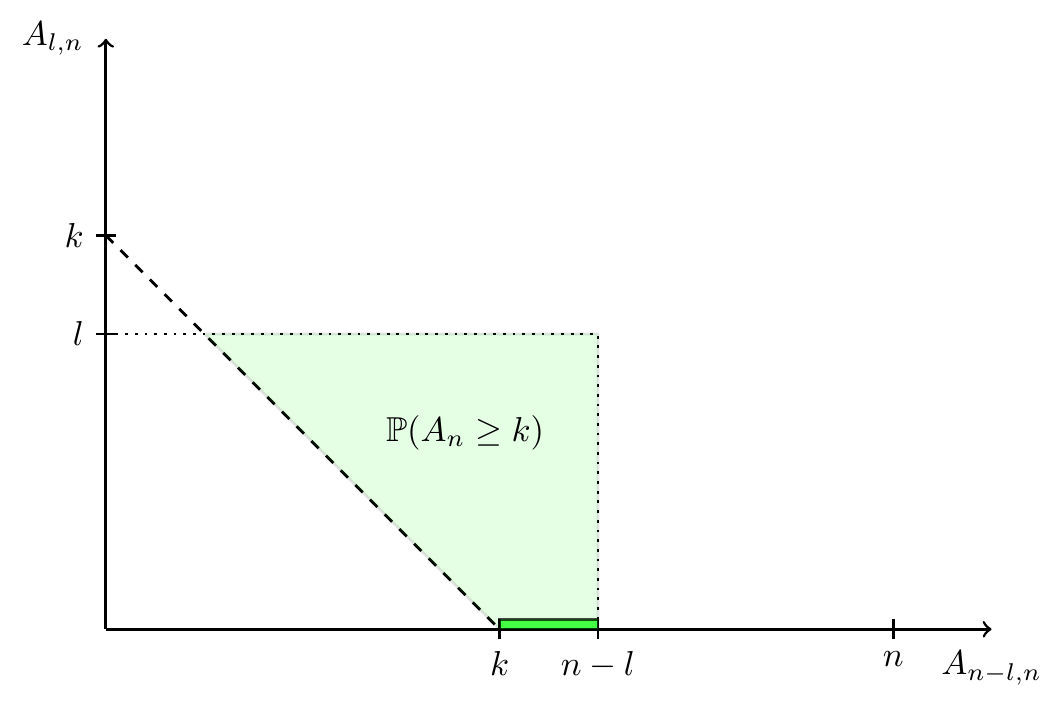}%
	}\hfill
\subfloat[Case $1-\beta \leq \alpha<1$.\label{fig:SumCaseC}]{%
  \includegraphics[width=.3\linewidth]{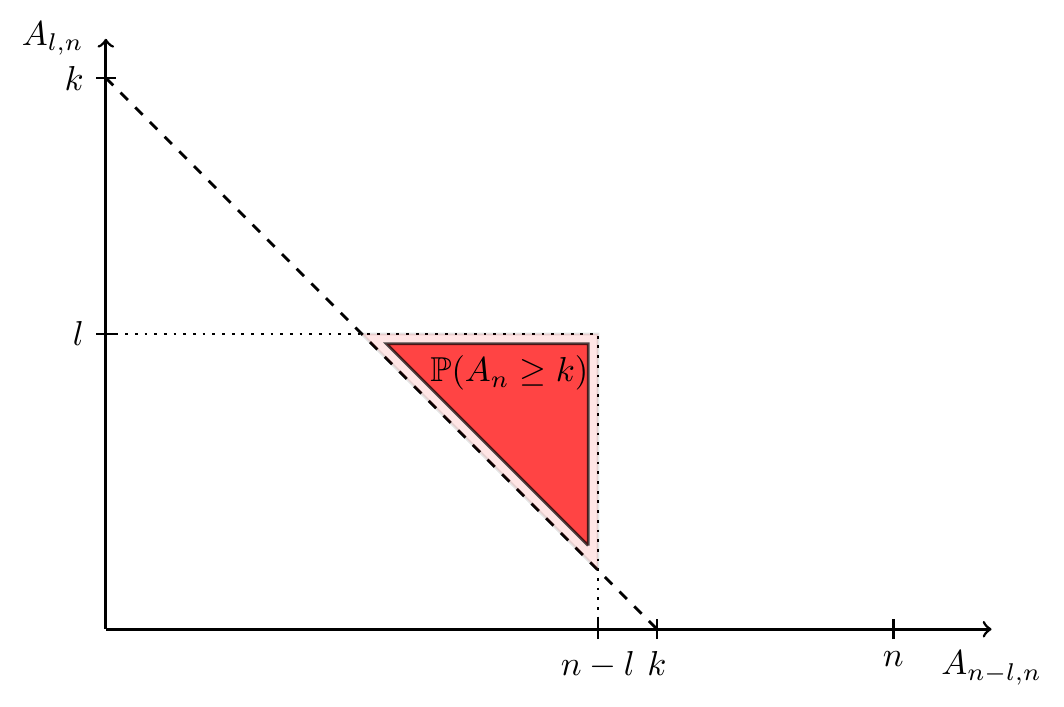}%
	}
\caption{Asymptotic contributions to the exceedance probability in Theorem~\ref{thm:ProportionalExceedanceNormal}.}
\label{fig:SumCases}
\end{figure*}

The cascading failure process is initiated by the failure of the single line connecting the two components. This event creates two disjoint components, and causes the load at each other line to increase by $\theta/n$ for a certain constant $\theta>0$. If this load increase exceeds the capacity of one or more lines, those lines will fail. Every subsequent failure again results in a load increase at the surviving lines, and we call such an increase the \emph{load surge}. This cascading failure process continues until the surplus capacity for every surviving line exceeds its load. We assume that the load surge caused by each consecutive failing line in the smaller component is $1/l$, and in the larger component $1/(n-l)$. We assume that both components remain connected after every consecutive line failure. In other words, the cascading failure processes behave independently between the two components and no further splitting will occur. 

The vulnerability of the network is measured by the probability that the blackout size, i.e.~the number of failed lines, after the cascading failure propagation has stopped, exceeds a certain threshold $k:=k_n$ as $n$ grows large. The asymptotic behavior shows how the \emph{exceedance probability} decays with respect to the threshold, and furnishes valuable qualitative insights. We consider all thresholds $k$ that are \emph{growing} with $n$, i.e.~both $k \rightarrow \infty$ and $n-k \rightarrow \infty$ as $n \rightarrow \infty$. 

Naturally, the behavior of the exceedance probability depends heavily on the sizes of the two components. We will consider the balanced case where both components have a size of order~$n$, as well as the disparate case where the smaller component is of a size smaller than order $n$.

Next, we introduce some notation that will be used throughout the paper. Let $A_n$ be the random variable representing the total number of line failures, and $A_{l,n}$ the number of line failures in a component of size $l$ disconnected from a component of size $n-l$. We assume both $\alpha:=\lim_{n\rightarrow \infty} k/n$ and $\beta:=\lim_{n\rightarrow \infty} l/n$ exist. We write $a_n=o(b_n)$ if $\lim_{n \rightarrow \infty} a_n/b_n =0$ and $a_n=O(b_n)$ if $\limsup_{n \rightarrow \infty} a_n/b_n< \infty$. Similarly, we write $a_n=\omega(b_n)$ if $\lim_{n \rightarrow \infty} b_n/a_n =0$ and $b_n=\Omega(b_n)$ if $\limsup_{n \rightarrow \infty} b_n/a_n< \infty$. Finally, we denote $a_n \sim b_n$ if $\lim_{n\rightarrow\infty} a_n/b_n=1$ and $a_n \propto b_n$ if $\lim_{n\rightarrow\infty} a_n/b_n \in (0,\infty)$.

The case of a single star network, where each line failures causes a single node to become isolated, has been studied rigorously in~\cite{Sloothaak2016}. Specifically, this case involves a star network consisting of $n+1$ nodes, $n$ lines with uniformly distributed surplus capacities, an initial load surge of $\theta/n$ at all lines, and subsequent load surges of $1/n$ at all surviving lines. In that case, the following result holds.

\begin{theorem}
Let $k_\star:={k_n}_\star$ and $k^\star:={k_n}^\star$ both be growing sequences of $n$ with $k_\star \leq k^\star$. Then,
\begin{align}
\lim_{n \rightarrow \infty} \sup_{k \in [k_\star,k^\star]} \big| k^{3/2} \sqrt{\frac{n-k}{n}} \Prob\left( A_n = k\right) - \frac{\theta}{\sqrt{2\pi}} \big| = 0,
\label{eq:NoSplittingEqual}
\end{align}
and
\begin{align}
\lim_{n \rightarrow \infty} \sup_{k \in [k_\star,k^\star]} \big| k^{1/2} \sqrt{\frac{n}{n-k}} \Prob\left( A_n \geq k\right) - \frac{2\theta}{\sqrt{2\pi}} \big| = 0.
\label{eq:NoSplittingExceedance}
\end{align}
\label{thm:NoSplitting}
\end{theorem}

\noindent
Theorem~\ref{thm:NoSplitting} thus states that uniformly for all $k \in [k_\star,k^\star]$,
\begin{align*}
\Prob\left( A_n = k\right) \sim \frac{\theta}{\sqrt{2\pi}} \sqrt{\frac{n}{n-k}}k^{-3/2},
\end{align*}
and
\begin{align*}
\Prob\left( A_n = k\right) \sim \frac{2\theta}{\sqrt{2\pi}} \sqrt{\frac{n-k}{n}} k^{-1/2}.
\end{align*}
Equation~\eqref{eq:NoSplittingEqual} is Theorem~1 of~\cite{Sloothaak2016}. The proof of~\eqref{eq:NoSplittingExceedance} can be found in the Appendix: it follows the lines of the proof of Theorem~2 in~\cite{Sloothaak2016}, but it is adapted to hold uniformly.

\section{Main results}
\label{sec:Main results}
The exceedance probability naturally depends on the threshold and the component sizes. In essence, we derive the tail distribution of $A_n=A_{l,n}+A_{n-l,n}$, where $A_{l,n}$ and $A_{n-l,n}$ are independent random variables. Note that $A_{l,n}$ involves the number of line failures in a single star network with initial load surge $\theta/n =\theta/l \cdot l/n$ and consecutive load surges $1/l$. Therefore, $A_{l,n}$ obeys a quasi-binomial distribution~\cite{Dobson2005}, where the asymptotic behavior is given by Theorem~\ref{thm:NoSplitting} (with $\theta$ replaced by $\theta \cdot l/n$). We point out that $A_{l,n}$ is thus heavy-tailed for all values that are not too close to $l$. We derive the asymptotic behavior of the probability that the sum of two quasi-binomial distributed random variables exceeds a network-size dependent threshold $k$.

As mentioned earlier, we distinguish between two cases: the balanced case where $\beta = \lim_{n \rightarrow \infty} l/n >0$, and the disparate case where $l=o(n)$.

\begin{table*}[htb]
\caption{Road map for proof of Theorem~\ref{thm:ProportionalExceedanceNormal}}
\begin{ruledtabular}
\begin{tabular}{cccc}
 & $\Prob\left(A_n \geq k ; A_{l,n} \leq s_\star \right)$ & $\Prob\left(A_n \geq k ; s_\star <  A_{l,n} < s^\star \right)$ & $\Prob\left(A_n \geq k ; A_{l,n} \geq  s^\star \right)$ \\
\hline
$0 \leq \alpha < \beta$ & $\sim \Prob\left(A_{n-l,n} \geq k \right)$ & negligible & $\sim \Prob\left(A_{l,n} \geq n\right)$ \\
$\beta \leq \alpha < 1-\beta$ & $\sim \Prob\left(A_{n-l,n} \geq  k \right)$ & negligible & 0 or negligible \\
$1-\beta < \alpha < 1$ & 0 & dominant & 0 \\
\end{tabular}
\end{ruledtabular}
\label{tab:OverviewNormalCase}
\end{table*}

\subsection{Balanced component sizes}
In this section we consider the case where the two component sizes are of the same order, and derive the tail of $A_n$. This tail behavior reflects the most likely scenarios for the number of line failures to exceed threshold $k$. Recall that $A_n$ is essentially the sum of two heavy-tailed random variables (when $l\rightarrow \infty$ as $n \rightarrow \infty$). Moreover, the tail of both random variables typically obeys a power-law distribution with exponent $-1/2$ in the balanced case. 

This observation yields an intuition for the asymptotic behavior of the exceedance probability. Figure~\ref{fig:SumCases} visually illustrates this intuition, where the thicker areas reflect which scenarios asymptotically contribute to the exceedance probability. When the threshold is significantly smaller than both component sizes ($\alpha < \beta$), the most likely scenario to exceed $k$ is when it is exceeded in one of the components alone. In other words, the event where both $A_{l,n}$ and $A_{n-l,n}$ attain large values is much less likely to occur. Similarly, if the threshold is only significantly smaller than the larger component size ($\beta \leq \alpha <1-\beta$), the most likely scenario for $A_n$ to exceed $k$ is when it is exceeded in the larger component, while the smaller component only has very few line failures. We observe that in both cases, the tail of $A_n$ therefore obeys a power-law distribution with exponent $-1/2$.

If $1-\beta<\alpha<1$, both components must have many line failures. The threshold is then most likely to be exceeded if in both components a non-negligible fraction of the lines have failed. This causes the power-law exponent to decrease to $-1/2 \cdot 1/2 =-1$, i.e.~a phase-transition appears at $\alpha=1-\beta$. Table~\ref{tab:OverviewNormalCase} summarizes these notions, providing a road map to prove the following theorem.

\begin{theorem}
Suppose $\beta \in (0,1/2]$ and $\alpha \neq 1-\beta$. As $n \rightarrow \infty$, $A_n$ asymptotically behaves as in Table~\ref{tab:ProportionalExceedanceNormalTheorem}, where
\begin{align*}
c(\alpha,\beta) &= \int_{x=\frac{\alpha-(1-\beta)}{\beta}}^{1}   \frac{x^{-3/2}}{\sqrt{1-x}} \sqrt{\frac{s(x)}{1-s(x)}} \, dx, \\
 s(x) &= \frac{\beta x - (\alpha-(1-\beta))}{1-\beta}.
\end{align*}
\label{thm:ProportionalExceedanceNormal}
\end{theorem}

\begin{table}[b]
\caption{\label{tab:ProportionalExceedanceNormalTheorem}%
Asymptotic behavior of $A_n$ if $\beta \in (0,1/2]$ and $\alpha \neq 1-\beta$.
}
\begin{ruledtabular}
\begin{tabular}{c|c}
If & As $n\rightarrow \infty$, $\Prob\left(A_n \geq k\right) \sim$ \\
\hline
$0 \leq \alpha<\beta$ 			& $\frac{2\beta \theta}{\sqrt{2\pi}} \sqrt{1-\frac{k}{l}} k^{-1/2} + \frac{2(1-\beta) \theta}{\sqrt{2\pi}} \sqrt{1-\frac{k}{n-l}} k^{-1/2}$\\
$\beta \leq \alpha < 1-\beta$ 	& $\frac{2(1-\beta) \theta}{\sqrt{2\pi}} \sqrt{1-\frac{k}{n-l}} k^{-1/2}$\\
$1-\beta < \alpha <1$			& $\frac{\alpha \sqrt{\beta(1-\beta)} \theta^2}{\pi} c(\alpha,\beta) k^{-1}$\\
\end{tabular}
\end{ruledtabular}
\end{table}

\begin{table}[b]
\caption{\label{tab:ProportionalExceedanceSpecialTheorem}%
Asymptotic behavior of $A_n$ as $n \rightarrow \infty$ if $\beta \in (0,1/2]$ and $\alpha = 1-\beta$, with $\eta := \lim_{n\rightarrow \infty} t/r$.
}
\begin{ruledtabular}
\begin{tabular}{c|c}
If & As $n\rightarrow \infty$, $\Prob\left(A_n \geq k\right) \sim$ \\
\hline
$\beta \in [0,1/2)$ & \\
\hline
$r<0, -r= \omega((\log k)^2)$ & $\frac{2(1-\beta)\theta}{\sqrt{2\pi}} \frac{\sqrt{-r}}{k}$\\
$r<0, -r \propto (\log k)^2$ & $\frac{2(1-\beta)\theta}{\sqrt{2\pi}} \frac{\sqrt{-r}}{k} + \frac{\beta(1-\beta) \theta^2}{\pi} \frac{\log k}{k}$\\
$r<0, -r \propto (\log k)^2$ & $\frac{\beta(1-\beta) \theta^2}{\pi} \frac{\log k}{k}$\\
otherwise & $\frac{\beta(1-\beta) \theta^2}{\pi} \frac{\log (k/r)}{k}$\\
\hline
$\beta =1/2$ & \\
\hline
$-r= \omega((\log k)^2), \eta >0$ & $\frac{\theta}{\sqrt{2\pi}} \frac{\sqrt{-r}+\sqrt{-t}}{k}$\\
$-r= \omega((\log k)^2), \eta \leq 0$ & $\frac{\theta}{\sqrt{2\pi}} \frac{\sqrt{-r}}{k}$\\
$-r \propto (\log k)^2, \eta >0$ & $\frac{\theta}{\sqrt{2\pi}} \frac{\sqrt{-r}+\sqrt{-t}}{k} +\frac{\theta^2}{2\pi} \frac{\log k}{k}$ \\
$-r \propto (\log k)^2, |t| \neq k^{1-o(1)}$ & $\frac{\theta}{\sqrt{2\pi}} \frac{\sqrt{-r}}{k} +\frac{\theta^2}{4\pi} \frac{\log k+\log(k/(|t|+1))}{k}$\\
$-r \propto (\log k)^2, \, t = k^{1-o(1)}$& $\frac{\theta}{\sqrt{2\pi}} \frac{\sqrt{-r}}{k} +\frac{\theta^2}{4\pi} \frac{\log k}{k}$ \\
$r \textrm{ else}, \frac{k}{|t|+1} \neq \left(\frac{k}{|r|+1}\right)^{o(1)}$ & $\frac{\theta^2}{4\pi} \frac{\log(k/(|r|+1))+\log(k/(|t|+1))}{k}$ \\
otherwise & $\frac{\theta^2}{4\pi} \frac{\log (k/(|r|+1))}{k}$\\
\end{tabular}
\end{ruledtabular}
\end{table}

\noindent
The proof of Theorem~\ref{thm:ProportionalExceedanceNormal} makes the provided intuition more rigorous. That is, we partition the event of exceeding the threshold in three terms:
\begin{align}
&\{A_n \geq k \} = \{A_n \geq k ; A_{l,n} \leq s_\star \} \nonumber \\
&\, \cup \{A_n \geq k ; s_\star <  A_{l,n} < s^\star \} \cup \{A_n \geq k ; A_{l,n} \geq  s^\star \}.
\end{align}
where $s_\star$ is chosen appropriately small, and $s^\star$ appropriately large. Table~\ref{tab:OverviewNormalCase} illustrates which term will yield the dominant behavior in each of the cases in Theorem~\ref{thm:ProportionalExceedanceNormal}.

The reasoning turns more subtle at the boundary where the threshold is either close to the larger component size, or when it close to $n$ itself. In view of Theorem~\ref{thm:ProportionalExceedanceNormal}, the first case ($\alpha=1-\beta$) corresponds to the interval of threshold values where we move from a power-law distribution with exponent $-1/2$ to one with exponent $-1$. When the larger component remains significantly larger than the smaller one ($0 <\beta <1-\beta <1$), this phase transition occurs as follows. As long as threshold $k$ is sufficiently smaller than $l$, the most likely scenario to exceed $k$ remains when it is already exceeded in the larger component alone. However, the closer $k$ is to $l$, the smaller this probability is and it is in fact zero when $l > k$. From some specific point, the scenario where the number of line failures in the larger component is close to $k$, yet not exceeding it, becomes the most likely one. If $\alpha=\beta=1-\beta=1/2$, a similar likely event can also occur for the smaller component. Figure~\ref{fig:SumCases2} reflects this intuition of Theorem~\ref{thm:ProportionalExceedanceSpecial}. Again, the thick areas indicate which scenarios possibly asymptotically contribute to the exceedance probability.

\begin{figure*}
\centering
\subfloat[Case $\beta \in (0,1/2)$.\label{fig:SumCaseA2}]{%
  \includegraphics[width=.45\linewidth]{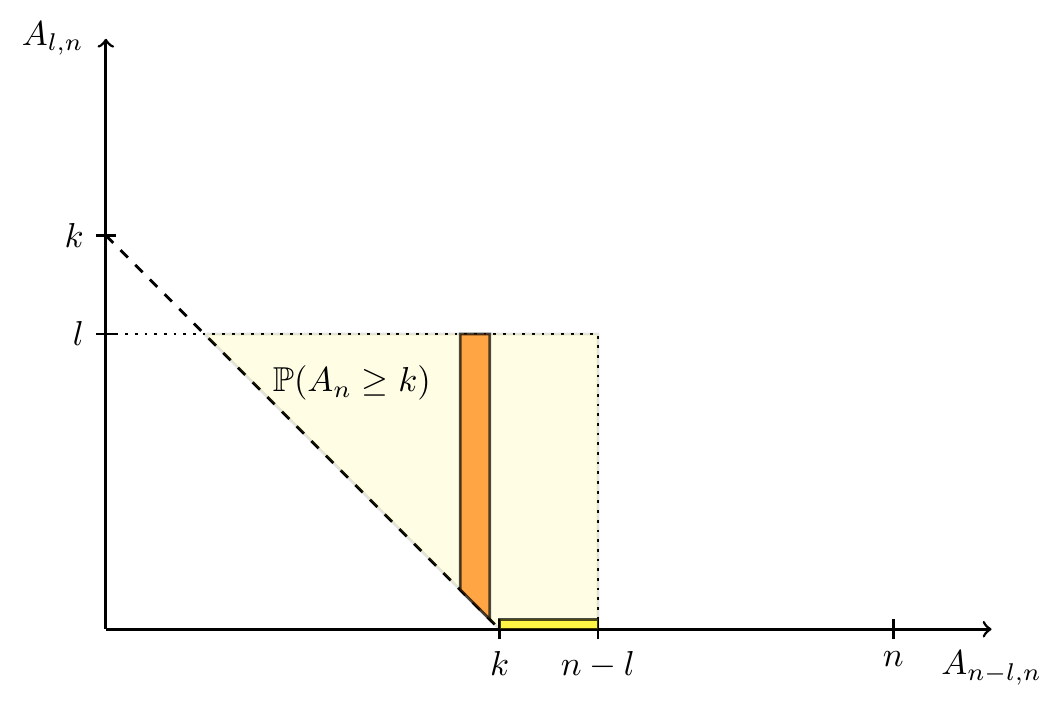}%
	}\hfill
\subfloat[Case $\beta=1-\beta=1/2$.\label{fig:SumCaseB2}]{%
  \includegraphics[width=.45\linewidth]{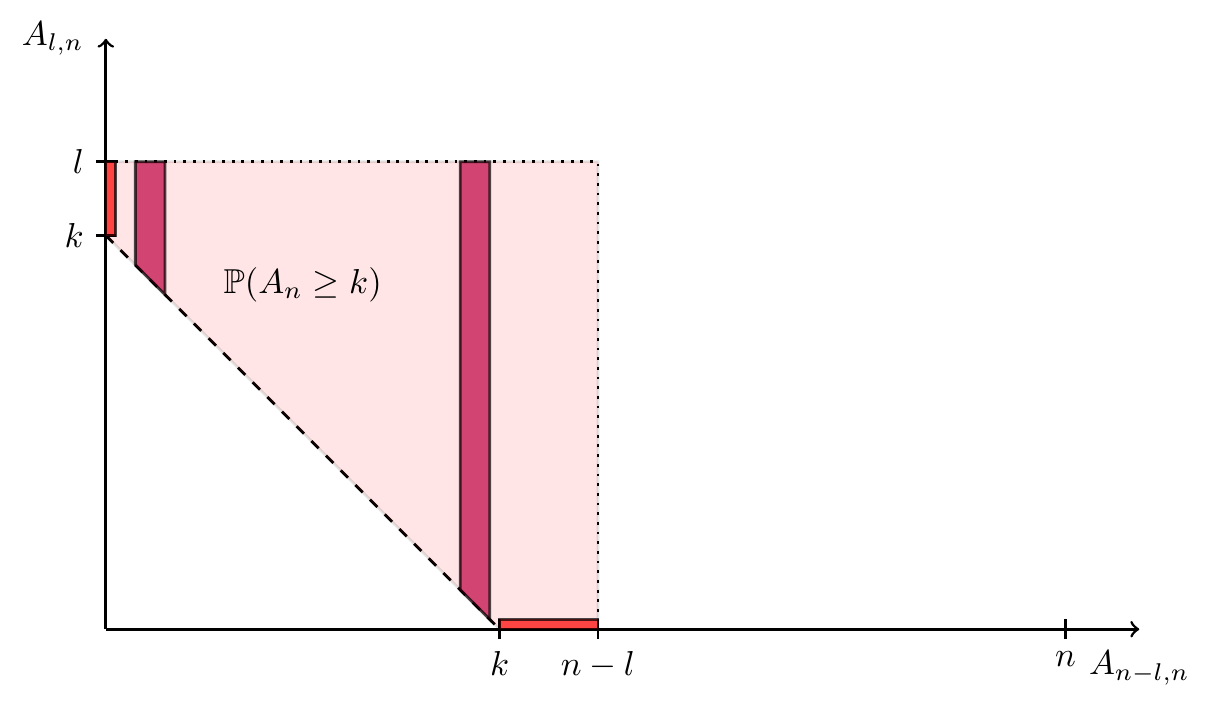}%
	}
\caption{Asymptotic contributions to $\Prob(A_n \geq k)$ in Theorem~\ref{thm:ProportionalExceedanceSpecial} if $\alpha=1-\beta$.}
\label{fig:SumCases2}
\end{figure*}

Finally, if the threshold is close to the network size itself ($\alpha=1$), almost all lines in both of the components need to have failed. Visually, this case is comparable to the one in Figure~\ref{fig:SumCaseC}, where the triangle is minuscule. 

\begin{theorem}
Suppose $\alpha=1-\beta$ and $\beta \in (0,1/2]$, and write $r:=r_n=k-(n-l)$ and $t:=t_n=k-l$. Then as $n \rightarrow \infty$, $A_n$ asymptotically behaves as in Table~\ref{tab:ProportionalExceedanceSpecialTheorem}. If $\alpha=1$, then as $n\rightarrow \infty$,
\begin{align}
\Prob\left(A_n \geq k\right) \sim \frac{\theta^2}{2}(n-k)k^{-2}.
\end{align}
\label{thm:ProportionalExceedanceSpecial}
\end{theorem}

\noindent
Since there is a sharp transition from a power-law with exponent $-1/2$ to one with exponent $-1$ when $\alpha=1-\beta$, it is natural to consider the number of failures in the bigger component in more detail. In the proof of Theorem~\ref{thm:ProportionalExceedanceSpecial}, we partition the event of exceeding the threshold with respect to the number of line failures in the bigger component. 

When $\alpha=1-\beta$ with $\beta \in (0,1/2)$, we use the identity
\begin{align}
\Prob&\left(A_n \geq k \right)  = \Prob\left(A_{n-l,n} \geq k \right) \nonumber\\
&\hspace{.6cm}+\Prob\left(A_n \geq k ; A_{n-l,n} \in [k-s_\star,k) \right) \nonumber\\
&\hspace{1cm}+\Prob\left(A_n \geq k ; A_{n-l,n} \in (k-s^\star,k-s_\star)\right) \nonumber\\
&\hspace{1.5cm}+\Prob\left(A_n \geq k ; A_{n-l,n} \in [k-l,k-s^\star] \right),
\label{eq:SpecialFirstIdentity}
\end{align}
where $s_\star$ and $s^\star$ are chosen in a specific way. Labeling each term I, II, III, IV respectively, the asymptotic behavior of each term can be evaluated separately, which yields the result as in Table~\ref{tab:TermsAsympoticsCase1}.

\begin{table*}[htb]
\begin{ruledtabular}
\begin{tabular}{ccc}
Term & Probability & Asymptotic behavior \\
\hline
I & $\Prob\left(A_{n-l,n} \geq k \right)$ & $\frac{2(1-\beta)\theta}{\sqrt{2\pi}} \frac{\sqrt{-r}}{k} \mathbbm{1}_{\{-r>0 \textrm{ growing}\}}+ O(k^{-1}) \cdot \mathbbm{1}_{\{r\leq 0 \textrm{ fixed}\}}$ \\

II & $\Prob\left(A_n \geq k ; A_{n-l,n} \in [k-s_\star,k) \right)$ & $o\left( \frac{\log(l/ (|r|+1))}{k}  \right)$ \\

III & $\Prob\left(A_n \geq k ; A_{n-l,n} \in (k-s^\star,k-s_\star) \right)$ & $\frac{\beta(1-\beta) \theta^2}{\pi} \frac{\log (k/(|r|+1)}{k}$\\

IV & $\Prob\left(A_n \geq k ; A_{n-l,n} \in [k-l,k-s^\star] \right)$ & $o\left( \frac{\log(l/ (|r|+1))}{k}  \right)$ \\
\end{tabular}
\end{ruledtabular}
\caption{Asymptotic behavior of terms in~\eqref{eq:SpecialFirstIdentity}.}
\label{tab:TermsAsympoticsCase1}
\end{table*}

The result then follows by determining the dominant terms of~\eqref{eq:SpecialFirstIdentity} for the various cases of the threshold. It turns out there is a transition in dominant behavior when $-r \propto (\log k)^2$. For a smaller threshold, the threshold remains most likely to be exceeded in the larger component alone. Otherwise, it is most likely that exceeding the threshold is caused by almost all lines failing in the larger component in conjunction with a growing number of line failures in the smaller component. This outcome is summarized by Table~\ref{tab:OverviewSpecialCase1}. 
\begin{table*}[htb]
\begin{ruledtabular}
\begin{tabular}{ccccc}
 & I & II & III & IV  \\
\hline
$-r = \omega((\log k)^2)$ & $\frac{2(1-\beta)\theta}{\sqrt{2\pi}} \frac{\sqrt{-r}}{k}$ & negligible & negligible & negligible \\

$-r \propto (\log k)^2$ & $\frac{2(1-\beta)\theta}{\sqrt{2\pi}} \frac{\sqrt{-r}}{k}$ & negligible & $\frac{\beta(1-\beta) \theta^2}{\pi} \frac{\log k}{k}$ & negligible \\

$|r| = o(\log k)^2$ & 0 or negligible & negligible & $\frac{\beta(1-\beta) \theta^2}{\pi} \frac{\log k}{k}$ & negligible \\ \\

$r>0$ growing & 0 & negligible & $\frac{\beta(1-\beta) \theta^2}{\pi} \frac{\log (k/r)}{k}$ & negligible \\
\end{tabular}
\end{ruledtabular}
\caption{Road map for proof of Theorem~\ref{thm:ProportionalExceedanceSpecial} with $\alpha=1-\beta$ with $\beta \in (0,1/2)$.}
\label{tab:OverviewSpecialCase1}
\end{table*}

The situation turns even more subtle when $\alpha=\beta=1/2$. In the most extreme case, we may have $l=n-l=n/2$ and one cannot distinguish between a smaller and larger component. The cascading process in the component of size $l$ can therefore become more significant, leading to more possible scenarios likely to have occurred if the threshold is exceeded. To capture these scenarios, we need to refine the partitioning of events in~\eqref{eq:SpecialFirstIdentity} to
\begin{align}
&\Prob\left(A_n \geq k \right)  = \Prob\left(A_{n-l,n} \geq k \right) \nonumber\\
&+ \Prob\left(A_n \geq k ; A_{n-l,n} \in [k-s_\star,k) \right) \nonumber\\
&\hspace{.3cm}+ \Prob\left(A_n \geq k ; A_{n-l,n} \in (k-s^\star,k-s_\star) \right) \nonumber\\
&\hspace{.6cm}+\Prob\left(A_n \geq k ; A_{n-l,n} \in (k-q_\star,k-s^\star] \right)\nonumber\\
&\hspace{.9cm}+\Prob\left(A_n \geq k ; A_{n-l,n} \in [k-q^\star,k-q_\star] \right)\nonumber\\
&\hspace{1.2cm}+\Prob\left(A_n \geq k ; A_{n-l,n} \in [k-l,k-q^\star) \right).
\label{eq:SpecialSecondIdentity}
\end{align}
In other words, we partition the event $\{A_n \geq k ; A_{n-l,n} \in [k-l,k-s^\star] \}$ in~\eqref{eq:SpecialFirstIdentity} in three disjoint events in this case. In the proof, we determine the asymptotic behavior of the various terms in the identity~\eqref{eq:SpecialSecondIdentity}, which leads to the result given in Table~\ref{tab:TermsAsympoticsCase2}. Table~\ref{tab:OverviewSpecialCase2} illustrates which terms contribute to the asymptotic tail behavior of $A_n$.

\begin{table*}[htb]
\begin{ruledtabular}
\begin{tabular}{ccc}
Term & Probability & Asymptotic behavior \\
\hline
I & $\Prob\left(A_{n-l,n} \geq k \right)$ & $\frac{\theta}{\sqrt{2\pi}} \frac{\sqrt{-r}}{k} \mathbbm{1}_{\{-r>0 \textrm{ growing}\}}+ O(k^{-1}) \cdot \mathbbm{1}_{\{r\leq 0 \textrm{ fixed}\}}$ \\

II & $\Prob\left(A_n \geq k ; A_{n-l,n} \in [k-s_\star,k) \right)$ & $o\left( \frac{\log(l/ (|r|+1))}{k}  \right)$ \\

III & $\Prob\left(A_n \geq k ; A_{n-l,n} \in (k-s^\star,k-s_\star) \right)$ & $\frac{\theta^2}{4\pi} \frac{\log k}{k}$\\

IV & $\Prob\left(A_n \geq k ; A_{n-l,n} \in (k-q_\star,k-s^\star] \right)$ & $o\left( \frac{\log(k/(|r|+1)}{k}+ \frac{\log(k/(|t|+1)}{k} \right)$ \\

V & $\Prob\left(A_n \geq k ; A_{n-l,n} \in [k-q^\star,k-q_\star] \right)$ & $\frac{\theta^2}{4\pi} \frac{\log(k/(|t|+1)}{k}$\\

VI & $\Prob\left(A_n \geq k ; A_{n-l,n} \in [k-l,k-q^\star) \right)$ & $\frac{\theta}{\sqrt{2\pi}} \frac{\sqrt{-t}}{k} \mathbbm{1}_{\{-t>0 \textrm{ growing}\}} + o\left( \frac{\log(l/ (|t|+1))}{k}  \right)$ 
\end{tabular}
\end{ruledtabular}
\caption{Asymptotic behavior of terms in~\eqref{eq:SpecialSecondIdentity}.}
\label{tab:TermsAsympoticsCase2}
\end{table*}

\begin{table*}[htb]
\begin{ruledtabular}
\begin{tabular}{ccccccc}
 & I & II & III & IV & V & VI \\
\hline
$-r= \omega((\log k)^2), \lim_{n\rightarrow \infty} t/r >0$ & $\frac{\theta}{\sqrt{2\pi}} \frac{\sqrt{-r}}{k}$ & negl. & negl. & negl. & negl. & $\frac{\theta}{\sqrt{2\pi}} \frac{\sqrt{-t}}{k}$ \\

$-r= \omega((\log k)^2), \, \lim_{n\rightarrow \infty} t/r \leq 0$ & $\frac{\theta}{\sqrt{2\pi}} \frac{\sqrt{-r}}{k}$ & negl. & negl. & negl. & negl. & negl. \\

$-r \propto (\log k)^2, \lim_{n\rightarrow \infty} t/r >0$ & $\frac{\theta}{\sqrt{2\pi}} \frac{\sqrt{-r}}{k}$ & negl. & $\frac{\theta^2}{4\pi} \frac{\log k}{k}$ & negl. & $\frac{\theta^2}{4\pi} \frac{\log k}{k}$ & $\frac{\theta}{\sqrt{2\pi}} \frac{\sqrt{-t}}{k}$ \\

$-r \propto (\log k)^2, |t| \neq k^{1-o(1)}$ & $\frac{\theta}{\sqrt{2\pi}} \frac{\sqrt{-r}}{k}$ & negl. & $\frac{\theta^2}{4\pi} \frac{\log k}{k}$ & negl. & $\frac{\theta^2}{4\pi} \frac{\log (\frac{k}{|t|+1})}{k}$ & negl. \\

$-r \propto (\log k)^2, \, t = k^{1-o(1)}$ & $\frac{\theta}{\sqrt{2\pi}} \frac{\sqrt{-r}}{k}$ & negl. & $\frac{\theta^2}{4\pi} \frac{\log k}{k}$ & negl. & negl. & negl. \\

$r \textrm{ otherwise}, \frac{k}{|t|+1} \neq \left(\frac{k}{|r|+1}\right)^{o(1)}$ & negl. & negl. & $\frac{\theta^2}{4\pi} \frac{\log (\frac{k}{|r|+1})}{k}$ & negl. & $\frac{\theta^2}{4\pi} \frac{\log (\frac{k}{|t|+1})}{k}$ & negl. \\

otherwise & negl. & negl. & $\frac{\theta^2}{4\pi} \frac{\log (\frac{k}{|r|+1})}{k}$ & negl. & negl. & negl. 
\end{tabular}
\end{ruledtabular}
\caption{Road map for proof of Theorem~\ref{thm:ProportionalExceedanceSpecial} with $\alpha=1-\beta=1/2$.}
\label{tab:OverviewSpecialCase2}
\end{table*}

The final case of Theorem~\ref{thm:ProportionalExceedanceSpecial} involves the case where the threshold is close to the network size, i.e.~$\alpha=1$. Both component sizes are therefore significantly smaller than the threshold. In this case, we partition the event of exceeding the threshold in only three disjoint events:
\begin{align}
\Prob&\left(A_n \geq k \right)  = \Prob\left(A_n \geq k ; A_{n-l,n} \in (k-s_\star,n-l] \right) \nonumber \\
&\hspace{.5cm}+ \Prob\left(A_n \geq k ; A_{n-l,n} \in [k-s^\star,k-s_\star] \right) \nonumber \\
&\hspace{1cm} + \Prob\left(A_n \geq k ; A_{n-l,n} \in (k-l,k-s^\star] \right).
\label{eq:SpecialThirdIdentity}
\end{align}
For appropriate choices of $s_\star$ and $s^\star$, we will show that the second term is dominant and yields the result in Theorem~\ref{thm:ProportionalExceedanceSpecial}.

\subsection{Disparate component sizes}
\label{sec:NonProportional}
Next, we turn to the case $l=o(n)$. The smaller component is hence of a size that is (almost) negligible compared to the larger component. Essentially, this results in a framework that for most thresholds ($0 < \alpha<1$), no matter what occurs in the smaller component, the only likely manner to exceed the threshold is when it is exceeded in the larger component alone. This intuition remains true for $\alpha = 0$: the initial disturbance $\theta/n = \theta/l \cdot l/n$ is relatively minor in the smaller component and unlikely to cause the cascading failure process to propagate further.

When $\alpha=1-\beta=1$, other scenarios to exceed $k$ may become relevant. In particular, when $k > n-l$ the number of line failures in the larger component alone cannot exceed $k$. The partitioning of the event of exceeding threshold $k$ needs to be done carefully, resulting in many phase transitions.

\begin{theorem}
Suppose $\beta =0$ and $r=k-(n-l)$. If $\alpha<1$, or $\alpha=1$ with $-r = \Omega(l)$, then as $n \rightarrow \infty$,
\begin{align}
\Prob\left(A_n \geq k\right) \sim \frac{2\theta}{\sqrt{2\pi}}  \sqrt{1-\frac{k}{n-l}} k^{-1/2}.
\label{eq:NonProportionalNormal}
\end{align}
If $k\leq n-l$, $-r = o(l)$ growing with $l$, then as $n \rightarrow \infty$,
\begin{align}
\Prob\left(A_n \geq k\right) \sim \left\{ \begin{array}{l}
\frac{2\theta}{\sqrt{2\pi}} \frac{\sqrt{-r}}{k}, \\
\frac{2\theta}{\sqrt{2\pi}} \frac{\sqrt{-r}}{k} + \frac{\theta^2 l \log l}{\pi} k^{-2}, \\
\frac{\theta^2 l \log l}{\pi} k^{-2} ,
\end{array}\right.
\label{eq:NonProportionalRVeryNegative}
\end{align}
if $l=o\left( \frac{n\sqrt{-r}}{\log n} \right)$, if $l\propto  \frac{n\sqrt{-r}}{\log n}$, or if $l = \omega\left( \frac{n\sqrt{-r}}{\log n} \right)$, respectively. If $r\leq 0$ is fixed, then as $n \rightarrow \infty$,
\begin{align}
\Prob\left(A_n \geq k\right) \sim \left\{ \begin{array}{l}
\sum_{m=0}^{\max\{ -r,\lfloor \theta \rfloor \}} \frac{\theta (m-\theta)^m}{m!} {e}^{-(m-\theta)} k^{-1},\\
\frac{\theta^2 l \log l}{\pi} k^{-2}\\
\hspace{.2cm}+\sum_{m=0}^{\max\{ -r,\lfloor \theta \rfloor \}} \frac{\theta (m-\theta)^m}{m!} {e}^{-(m-\theta)} k^{-1}, \\
\frac{\theta^2 l \log l}{\pi} k^{-2},
\end{array}\right.
\label{eq:NonProportionalRFixed}
\end{align}
if $l=o\left( \frac{n}{\log n} \right)$, if $l\propto  \frac{n}{\log n}$, or if $l = \omega\left( \frac{n}{\log n} \right)$, respectively. If $k>n-l$ and $r=o(l)$ growing with $l$, then as $n \rightarrow \infty$,
\begin{align}
\Prob\left(A_n \geq k\right) \sim \frac{\theta^2 l \log (l/r)}{\pi} k^{-2}.
\label{eq:NonProportionalRPositive}
\end{align}
If $k>n-l$ and $\gamma:= \lim_{n\rightarrow \infty} r/l \in (0,1)$, then as $n \rightarrow \infty$,
\begin{align}
\Prob\left(A_n \geq k\right) \sim \frac{\theta^2}{\pi} c(y) \frac{\sqrt{l}}{k^2},
\end{align}
where 
\begin{align*}
c(y)= \int_{y=\gamma}^1 \sqrt{\frac{1-y}{(y-\gamma)y}} \, dy .
\end{align*}
Finally, if $k>n-l$ and $r=l-o(l)$, then as $n \rightarrow \infty$,
\begin{align}
\Prob\left(A_n \geq k\right) \sim \frac{\theta^2}{2}(n-k)k^{-2}.
\end{align}
\label{thm:NonProportional}
\end{theorem}

\noindent 
The proof of Theorem~\ref{thm:NonProportional} is analogous to the proofs of Theorems~\ref{thm:ProportionalExceedanceNormal} and~\ref{thm:ProportionalExceedanceSpecial} (excluding the case $\alpha=\beta=1/2$). Yet, we need to account for the disparity in component sizes, changing the points where the phase transitions occur.

\section{Future research directions}
\label{sec:Future}
The results of Theorems~\ref{thm:ProportionalExceedanceNormal}-\ref{thm:NonProportional} identify how the power-law exponent and its prefactor are affected when a single immediate split occurs. Allowing the network split to occur after several line failures would give the full picture of the impact of a single split. We would need to account for the effects of earlier line failures, and specify the load surge after the network split. To tackle this problem requires a fundamentally different approach, and we intend to pursue this challenge in future work. 

Understanding non-immediate splitting is crucial to explore more involved splitting mechanisms. Investigation of the failure behavior for more general network structures, where load increments may even depend on spatial correlations, poses many interesting challenges for future research.

\appendix
\section{Proof of Theorem~\ref{thm:NoSplitting}}
\begin{proof}[Proof of Theorem~\ref{thm:NoSplitting}]
Equation~\eqref{eq:NoSplittingEqual} is Theorem~1 of~\cite{Sloothaak2016}. The second statement follows the lines of the proof of Theorem~2 of~\cite{Sloothaak2016}, but it is adapted here to show the slightly more general result.

Choose $\tilde{k}=n-\log(n-k^\star)$ and fix $\epsilon >0$. Following the proof of Theorem~2 in~\cite{Sloothaak2016}, we observe that for every $k \in [k_\star,k^\star]$ that
\begin{align*}
&\sqrt{\frac{kn}{n-k}} \Prob\left( A_n \geq k\right) \leq e \frac{\theta}{\theta/k+1} \frac{n}{k(n-k)} \\
&\hspace{1cm} + (1+\epsilon) \frac{2\theta}{\sqrt{2\pi}}+ c \cdot \frac{\sqrt{k/n}}{1-\log(n-\tilde{k})/n} \frac{\log(n-\tilde{k})}{\sqrt{n-k}}
\end{align*}
for some positive constant $c$, and
\begin{align*}
\sqrt{\frac{kn}{n-k}} \Prob\left( A_n \geq k\right) \geq  (1-\epsilon) \frac{2\theta}{\sqrt{2 \pi}}\left(1-\sqrt{\frac{k}{\tilde{k}}}\sqrt{\frac{n-\tilde{k}}{n-k}} \right)
\end{align*}
for large enough $n$. Therefore,
\begin{align*}
&\sup_{k \in [k_\star,k^\star]} \bigg| \sqrt{\frac{k n}{n-k}} \Prob(A_n \geq k) - \frac{2\theta}{\sqrt{2\pi}} \bigg| \\
&\leq \sup_{k \in [k_\star,k^\star]} \hspace{-.1cm} \max\left\{ \epsilon \frac{2\theta}{\sqrt{2 \pi}}\left(1-\sqrt{\frac{k}{\tilde{k}}}\sqrt{\frac{n-\tilde{k}}{n-k}} \right), \epsilon \frac{2\theta}{\sqrt{2\pi}} \right.\\
&  \hspace{.25cm} \left.+ e \frac{\theta}{\theta/k+1} \frac{n}{k(n-k)} +  \frac{c \cdot \sqrt{k/n}}{1-\log(n-\tilde{k})/n} \frac{\log(n-\tilde{k})}{\sqrt{n-k}} \right\}.
\end{align*}
We see that as $n \rightarrow \infty$, this gives
\begin{align*}
\lim_{n \rightarrow \infty} \sup_{k \in [k_\star,k^\star]} \bigg| \sqrt{\frac{k n}{n-k}}\Prob(A_n \geq k) - \frac{2\theta}{\sqrt{2\pi}} \bigg| \leq \epsilon \frac{2\theta}{\sqrt{2\pi}}.
\end{align*}
Letting $\epsilon \downarrow 0$ concludes the proof.
\end{proof}

\section{Proofs of main results}
\label{sec:ProofsOfMainResults}
In order to prove Theorems~\ref{thm:ProportionalExceedanceNormal}-\ref{thm:NonProportional}, we use an asymptotic analysis for the sum of two independent heavy-tailed random variables. Specifically, we determine the asymptotic tail of $A_n=A_{l,n}+A_{n-l,n}$, where $A_{l,n}$ and $A_{n-l,n}$ are independent. In fact, the distribution of $A_{l,n}$ and $A_{n-l,n}$ is well-understood due to the following observation. The first line failure disconnects the network in two separate star networks of sizes $l$ and $n-l$. For the first component, this causes an initial disturbance of $\theta/n = (\theta/l) \cdot (l/n)$, and every consecutive line failure causes an additional load surge of $1/l$. Hence, it falls in the framework of the model studied in~\cite{Sloothaak2016}, where $l$ is the network size and initial disturbance constant $\theta \cdot l/n$. It implies that for every growing $k$,
\begin{align}
\Prob\left(A_{l,n} \geq k \right) \left\{ \begin{array}{ll}
=0 & \textrm{if } k > l, \\
\sim \chi(l,m) \cdot l^{-1} & \textrm{if } l-k \geq 0 \textrm{ fixed,} \\
\sim \frac{2 \theta \cdot l/n}{\sqrt{2\pi}} \sqrt{\frac{l-k}{k l}} & \textrm{otherwise,}
\end{array} \right.
\label{eq:ExceedanceInOneComponent}
\end{align}
where the latter holds uniformly in accordance with Theorem~\ref{thm:NoSplitting}, and 
\begin{align*}
\chi(l,m)= \sum_{m=0}^{\max\{ l-k,\lfloor\theta \cdot l/n \rfloor \}} \frac{\theta \cdot l/n (m-\theta\cdot l/n)^m}{m!} e^{-(m-\theta \frac{l}{n})}.
\end{align*} 
Similarly, \eqref{eq:ExceedanceInOneComponent} holds for $A_{n-l,n}$ with $l$ replaced by $n-l$. We see that it implies a certain order of magnitude, i.e.
\begin{align*}
\Prob\left(A_{l,n} \geq k \right) = \left\{ \begin{array}{ll}
0 & \textrm{if } k > l, \\
O(\frac{l}{n} l^{-1}) & \textrm{if } l-k \geq 0 \textrm{ fixed,} \\
O\left( \frac{l}{n} \frac{\sqrt{l-k}}{l}\right) & \textrm{if } l-k=o(l) \textrm{ growing}, \\
O\left( \frac{l}{n} \sqrt{\frac{l-k}{k l}} \right) & \textrm{if } \lim_{n \rightarrow \infty} k/l \in (0,1),\\
O\left( \frac{l}{n} k^{-1/2} \right) & \textrm{if } k=o(l),
\end{array} \right.
\end{align*}
and again, similarly for $A_{n-l,n}$ with $l$ replaced by $n-l$. This will be used extensively throughout the proofs. 

To derive the main results stated in Section~\ref{sec:Main results}, we determine which scenarios are likeliest to cause $A_n=A_{l,n}+A_{n-l,n}$ to exceed the threshold.

\subsection{Very few or many failures in one component}
The strategy in all our proofs involves an appropriate partition of the event of exceeding the threshold. The partition is done by considering the joint event of exceeding the threshold, and having a specific number of failures in only one component. In this section, we state results on the asymptotic behavior of such joint events where there are very few or many line failures in one component. 

\subsubsection{Smaller component}
The proof of Theorem~\ref{thm:ProportionalExceedanceNormal} partitions the event of exceeding threshold $k$ in joint events where the number of failures in the smaller component is in a certain interval. The next two lemmas quantify the probability of $\{A_n \geq k ; A_{l,n} \leq s_\star\}$ with $s_\star$ very small, $\{A_n \geq k; A_{l,n} \geq s^\star\}$ with $s^\star$ very large. The proofs are given in Appendix~\ref{app:FewManyProofs}.

\begin{lemma}
Let $s_\star=o(\min\{k,l\})$ be growing. Then, as $n \rightarrow \infty$,
\begin{align*}
\Prob\left(A_n \geq k ; A_{l,n} \leq s_\star \right) \left\{ \begin{array}{ll}
\sim \Prob \left( A_{n-l,n} \geq k \right) & \textrm{if } \alpha < 1-\beta,\\
=0 & \hspace{-.6cm} \textrm{if } 1-\beta < \alpha <1.
\end{array} \right.
\end{align*}
\label{lem:ConditionSmallerComponentSmall}
\end{lemma}

\begin{lemma}
Let $s^\star$ be growing defined by $s^\star=k-o(k)$ if $k<l$ and $s^\star = l -o(l)$ otherwise. Then, as $n \rightarrow \infty$,
\begin{align*}
\Prob\left(A_n \geq k ; A_{l,n} \geq s^\star \right) \left\{ \begin{array}{lr}
\sim \Prob \left( A_{l,n} \geq k \right) & \textrm{if } \alpha < \beta,\\
=o(k^{-1/2}) & \hspace{-.3cm} \textrm{if } \beta \leq \alpha <1-\beta,\\
=o(k^{-1}) & \hspace{-.25cm} \textrm{if } 1-\beta < \alpha <1.
\end{array} \right.
\end{align*}
\label{lem:ConditionSmallerComponentLarge}
\end{lemma}

\subsubsection{Larger component: balanced case}
The proof of Theorem~\ref{thm:ProportionalExceedanceSpecial} partitions the event of exceeding threshold $k$ in joint events where the number of failures in the larger component is in a certain interval. The next lemma shows the asymptotic behavior where almost all lines in the larger component have failed.
\begin{lemma}
Suppose $\alpha=1-\beta$ with $\beta \neq 0$. If $r<0$ and $-r$ is growing with $n$ or $|r|$ fixed, let $s_\star=o(l)$ be growing. Then,
\begin{align}
\Prob&\left(A_n \geq k, A_{n-l,n} \in [k-s_\star,\min\{k-1,n-l\}] \right) \nonumber \\
& \hspace{5cm}  = O\left(\frac{s_\star}{n-l} \right).
\end{align}
If $\alpha = 1-\beta$ and $r>0$ growing, let $s_\star=r+o(l)$ be growing such that $s_\star - r$ is growing. Then,
\begin{align}
\Prob\left(A_n \geq k, A_{n-l,n} \geq k-s_\star \right) = o\left(\frac{\sqrt{s_\star-r}}{n-l} \right).
\label{eq:ConditionBiggerComponentSmall}
\end{align}
Finally, if $\alpha=1$ and $\beta \neq 0$, let $s_\star=k-(n-l)+o(n-k)=r+o(n-k)$ such that $s_\star-r$ is growing. Then,
\begin{align}
\Prob\left(A_n \geq k, A_{n-l,n} \geq k-s_\star \right)  = o\left(\frac{n-k}{k^2} \right).
\label{eq:ConditionBiggerComponentSmallAlpha1}
\end{align}
\label{lem:ConditionBiggerComponentSmall}
\end{lemma}

\noindent
The second lemma in this section yields the asymptotic behavior where at least a significant number of lines in the smaller component have failed for $\alpha=1-\beta$ or $\alpha=1$.
\begin{lemma}
If $\beta \in (0,1/2)$ and $\alpha=1-\beta$, suppose $s^\star = o(l)$ such that $s^\star-r>0$ is growing. Then,
\begin{align}
\Prob\left(A_n \geq k, A_{n-l,n} \in [k-l,k-s^\star] \right) = O\left({s^\star}^{-1/2} k^{-1/2} \right).
\label{eq:ConditionBiggerComponentLarge1}
\end{align}
If $\alpha=\beta=1-\beta=1/2$ and $l-k$ not growing, set $s^\star=l-o(l)$ such that $l-s^\star$ is growing. Then,
\begin{align}
\Prob\left(A_n \geq k, A_{n-l,n} \in [k-l,k-s^\star] \right) = O\left(\frac{\sqrt{l-s^\star}}{k} \right).
\label{eq:ConditionBiggerComponentLarge2}
\end{align}
If $\alpha=\beta=1-\beta=1/2$ and $l-k>0$ growing, set $s^\star=k-o(l-k)$ such that $k-s^\star$ is growing. Then as $n\rightarrow \infty$,
\begin{align}
\Prob\left(A_n \geq k, A_{n-l,n} \in [k-l,k-s^\star] \right) \sim \frac{\theta}{\sqrt{2\pi}} \frac{\sqrt{l-k}}{k}.
\label{eq:ConditionBiggerComponentLarge3}
\end{align}
Finally, if $\alpha=1$ and $\beta \in (0,1/2]$, set $s^\star=l-o(n-k)$. Then,
\begin{align}
\Prob\left(A_n \geq k, A_{n-l,n} \in [k-l,k-s^\star] \right) = o\left( \frac{n-k}{k^2}\right).
\label{eq:ConditionBiggerComponentLarge4}
\end{align}
\label{lem:ConditionBiggerComponentLarge}
\end{lemma}

\subsubsection{Larger component: disparate case}
We note that we consider only $\alpha=1-\beta=1$ with $\beta=0$, and hence $r=o(l)$.

\begin{lemma}
Suppose $|r|=o(l)$ and let $s_\star=o(l)$ be growing with $l$ such that it satisfies $s_\star = o(\log l)$ when $-r > 0$ growing. Otherwise $s_\star=r+v$ with $v$ growing with $l$ satisfying $v=o(\log (l/(|r|+1)))$. Then,
\begin{align*}
\Prob&\left(A_n \geq k ; A_{n-l,n} \in [k-s_\star,k-1] \right) \\
&\hspace{2cm}= o \left( \mathbbm{1}_{\{r \leq 0\}} \frac{l\log l}{k^2} +\mathbbm{1}_{\{r>0\}} \frac{l \log(l/r)}{k^2} \right).
\end{align*}
\label{lem:DisparateSSmall}
\end{lemma}

\begin{lemma}
Suppose $|r|=o(l)$, and let $s^\star=o(l)$ be growing with $l$ such that it satisfies $s^\star=\omega(l/(\log l)^2)$ when $r\leq 0$, and $s^\star=\omega(l/(\log(l/r))^2)$ when $r>0$. Then,
\begin{align*}
\Prob&\left(A_n \geq k; A_{n-l,n} \in [k-l,k-s^\star] \right)\\
&\hspace{1.5cm}= o \left( \mathbbm{1}_{\{r \leq 0\}}\left( \frac{l \log l}{k^2} \right)+\mathbbm{1}_{\{r>0\}} \frac{l \log(l/r)}{k^2} \right).
\end{align*}
\label{lem:DisparateSLarge}
\end{lemma}

\subsection{Proof of Theorem~\ref{thm:ProportionalExceedanceNormal}}
Next, we prove Theorem~\ref{thm:ProportionalExceedanceNormal} using the approach outlined in Table~\ref{tab:OverviewNormalCase}.

\noindent
\begin{proof}[Proof of Theorem~\ref{thm:ProportionalExceedanceNormal}]
If $\alpha < \beta$, set $s_\star=o(k)$ and $s^\star=k-\omega(k/s_\star)$ both growing large with $n$. Then,
\begin{align*}
\Prob&\left(A_n \geq k ; s_\star <  A_{l,n} < s^\star \right) \\
&\hspace{.75cm}\leq \underbrace{\Prob\left(A_{l,n} \geq s_\star \right)}_{=O(s^{-1/2})} \underbrace{\Prob\left(A_{n-l,n} \geq k-s^\star \right)}_{=o((k/s_\star)^{-1/2})} = o(k^{-1/2}).
\end{align*}
Applying Lemmas~\ref{lem:ConditionSmallerComponentSmall} and~\ref{lem:ConditionSmallerComponentLarge}, together with~\eqref{eq:ExceedanceInOneComponent}, yields
\begin{align*}
\Prob&\left(A_n\geq k \right) \sim \Prob\left(A_{n-l,n}\geq k \right) + \Prob\left(A_{l,n}\geq k \right) \\
&\sim \frac{2\beta \theta}{\sqrt{2\pi}} \sqrt{1-\frac{k}{l}} k^{-1/2}  + \frac{2(1-\beta) \theta}{\sqrt{2\pi}} \sqrt{1-\frac{k}{n-l}} k^{-1/2} 
\end{align*}
as $n \rightarrow \infty$.

If~$\beta \leq \alpha < 1-\beta$, set $s_\star=o(\min\{k,l\})$ and $s^\star=\min\{k,l\}-o(\min\{k,l\})$, both growing with $n$. Lemmas~\ref{lem:ConditionSmallerComponentSmall} and~\ref{lem:ConditionSmallerComponentLarge} imply that as $n \rightarrow \infty$,
\begin{align*}
\Prob\left(A_{l,n}\geq k \right) + \Prob&\left(A_{n-l,n}\geq k \right) \sim \Prob\left(A_{n-l,n}\geq k \right) \\
&\hspace{1cm}\sim \frac{2(1-\beta) \theta}{\sqrt{2\pi}} \sqrt{1-\frac{k}{n-l}} k^{-1/2}.
\end{align*}
The result follows in this case as well since
\begin{align*}
\Prob&\left(A_n \geq k ; s_\star <  A_{l,n} < s^\star \right)\\
&\hspace{1cm}\leq  \underbrace{\Prob\left(A_{l,n} \geq s_\star \right)}_{=o(1)} \underbrace{\Prob\left(A_{n-l,n} \geq k-l \right)}_{=O(k^{-1/2})} = o(k^{-1/2}).
\end{align*}

Finally, we consider the case $1-\beta<\alpha<1$. For every $s_\star=o(l)$, it holds that $\Prob(A_n \geq k ; A_{l,n} < s_\star)=0$. Lemma~\ref{lem:ConditionSmallerComponentLarge} implies for $s^\star=l-o(l)$,
\begin{align*}
\Prob(A_n \geq k ; A_{l,n} \geq s^\star) = o(k^{-1}).
\end{align*}

In addition, we have for $q^\star = n-l-o(n-l)$
\begin{align*}
\Prob&\left(A_n \geq k ; A_{n-l,n} \geq q^\star \right) \\
&\leq \underbrace{\Prob\left(A_{l,n}\geq k-(n-l) \right)}_{=O(k^{-1/2})} \underbrace{\Prob\left(A_{n-l,n} \geq q^\star \right)}_{=O(\sqrt{q^\star}/k))} = o(k^{-1}).
\end{align*}

\noindent
Therefore, it remains to be shown that for these choices of $q^\star$ and $s^\star$,
\begin{align}
\Prob&(A_n \geq k ; A_{l,n} < s^\star ; A_{n-l,n} < q^\star) \nonumber \\
&\hspace{3cm}\sim \frac{\alpha \sqrt{\beta(1-\beta)} \theta^2}{\pi} c(\alpha,\beta) k^{-1}
\label{eq:ToShowBalanced2}
\end{align}
as $n\rightarrow \infty$, where $c(\alpha,\beta)$ is defined as in the theorem. 

\begin{remark}\normalfont
We note that $c(\alpha,\beta)$ is a positive finite constant. That is, the function within the integral is non-negative and has a positive mass over the interval we integrate, and hence it is positive. Moreover, since $x/(1-x)$ is a non-negative increasing function for all $x \in [0,1)$ and $s(\cdot)$ is a linearly increasing function,
\begin{align*}
\int_{x=\frac{\alpha-(1-\beta)}{\beta}}^{1} & \left( 1-x\right)^{-1/2} x^{-3/2} \sqrt{\frac{s(x)}{1-s(x)}} \, dx \\
&\hspace{.5cm}\leq 2 \sqrt{\frac{1-\alpha}{\alpha-(1-\beta)}} \sqrt{\frac{s(1)}{1-s(1)}} <\infty.
\end{align*}
Indeed, the integral expression is a positive finite constant. 
\end{remark}

\noindent
Theorem~\ref{thm:NoSplitting} yields as $n \rightarrow \infty$,
\begin{align*}
\Prob&(A_n \geq k ; A_{l,n} < s^\star ; A_{n-l,n} < q^\star) \sim \frac{\beta(1-\beta)\theta^2}{2 \pi} \\
& \cdot \sum_{j=k-{q^\star}+1}^{s^\star-1} \sum_{m=(n-l)-{q^\star}+1}^{j-(k-(n-l))} \sqrt{\frac{l}{(l-j)j^3}} \sqrt{\frac{n-l}{m(n-l-m)^3}}.
\end{align*}

\noindent
An upper bound for the summation term is given by
\begin{align*}
&\sum_{j=k-{q^\star}+1}^{s^\star-1} \sum_{m=(n-l)-{q^\star}+1}^{j-(k-(n-l))} \sqrt{\frac{l}{(l-j)j^3}} \sqrt{\frac{n-l}{m(n-l-m)^3}}\\
&\leq  \int_{j=k-{q^\star}}^{s^\star} \hspace{-.1cm}\int_{m=(n-l)-{q^\star}}^{j-(k-(n-l))} \hspace{-.1cm} \sqrt{\frac{l\cdot (n-l)}{(l-j)j^3 m(n-l-m)^3}} \, dm \, dj \\
&\leq \frac{1}{\sqrt{l(n-l)}}\int_{x=\frac{k-(n-l)}{l}}^{1} \int_{y=0}^{\frac{l x-(k-(n-l))}{n-l}} \\
&\hspace{1.5cm} \left( 1-x\right)^{-1/2} x^{-3/2} \cdot y^{-1/2} (1-y)^{-3/2} \, dy \, dx \\
&= \frac{2}{\sqrt{l(n-l)}}\int_{x=\frac{k-(n-l)}{l}}^{1}  \left( 1-x\right)^{-1/2} x^{-3/2}\\
&\hspace{2cm}\cdot \sqrt{\frac{(l x-(k-(n-l)))/(n-l)}{1-(l x-(k-(n-l)))/(n-l)}} \, dx.
\end{align*}
Similarly, a lower bound is given by
\begin{align*}
&\sum_{j=k-{q^\star}+1}^{s^\star-1} \sum_{m=(n-l)-{q^\star}+1}^{j-(k-(n-l))} \sqrt{\frac{l}{(l-j)j^3}} \sqrt{\frac{n-l}{m(n-l-m)^3}}\\
&\geq \frac{1}{\sqrt{l(n-l)}}\int_{x=\frac{k-{q^\star}+2}{l}}^{\frac{s^\star-2}{l}} \int_{y=\frac{n-l-{q^\star}+2}{n-l}}^{\frac{l x-(k-(n-l))}{n-l}} \left( 1-x\right)^{-1/2} x^{-3/2}\\
&\hspace{3.5cm}\cdot    y^{-1/2} (1-y)^{-3/2} \, dy \, dx. 
\end{align*}
Due to our choices of $q^\star$ and $s^\star$, then as $n \rightarrow \infty$, the two integral expressions converge to the same constant. That is, 
\begin{align*}
\Prob&(A_n \geq k ; A_{l,n} < s^\star ; A_{n-l,n} < q^\star) \sim \frac{\beta(1-\beta)\theta^2}{2 \pi} \\
&\hspace{1.5cm}\cdot\frac{2}{\sqrt{l(n-l)}} \int_{x=\frac{\alpha-(1-\beta)}{\beta}}^{1}   \frac{x^{-3/2}}{\sqrt{1-x}} \sqrt{\frac{s(x)}{1-s(x)}} \, dx
\end{align*}
as $n \rightarrow \infty$, which asymptotically coincides with~\eqref{eq:ToShowBalanced2}.
\end{proof}

\subsection{Proof of Theorem~\ref{thm:ProportionalExceedanceSpecial}}
\begin{proof}[Proof of Theorem~\ref{thm:ProportionalExceedanceSpecial}]
First consider the case that $\beta \in (0,1/2)$. Using identity~\eqref{eq:SpecialFirstIdentity}, we observe that it suffices to show that the asymptotic behavior provided in Table~\ref{tab:TermsAsympoticsCase1} holds. That is, the result is immediate from Table~\ref{tab:OverviewSpecialCase1}, which in turn only highlights the dominant terms of Table~\ref{tab:TermsAsympoticsCase1}.

Let $s_\star=o(\log (l/(|r|+1)) )$ growing if $-r>0$ growing or $|r|$ fixed, and $s_\star = r+ o(\min\{r,\log(l/r)\})$ with $s^\star-r$ growing if $r>0$ growing. Let $s^\star=o(l)$ growing such that $s^\star=\omega \left( l/ \log(l/(|r|+1))\right)$. Note that due to this choice, $s^\star > s_\star$ for all $n$ large enough. Then Lemmas~\ref{lem:ConditionBiggerComponentSmall} and~\ref{lem:ConditionBiggerComponentLarge} yield
\begin{align*}
\sum_{j=1}^{s_\star} \Prob\left(A_{l,n} \geq j \right) \Prob\left(A_{n-l,n} =k-j \right) =  o\left( \frac{\log\left(\frac{l}{ |r|+1}\right)}{k}  \right),\\
\sum_{j=s^\star}^{l} \Prob\left(A_{l,n} \geq j \right) \Prob\left(A_{n-l,n} =k-j \right) = o\left( \frac{\log\left(\frac{l}{|r|+1}\right)}{k}  \right).
\end{align*}

\noindent
Moreover, uniformly as $n \rightarrow \infty$,
\begin{align*}
&\sum_{j=s_\star}^{s^\star} \Prob\left(A_{l,n} \geq j \right) \Prob\left(A_{n-l,n} =k-j \right) \\
&\hspace{2.5cm} \sim \frac{\beta(1-\beta)\theta^2}{\pi} \sum_{j=s_\star}^{s^\star} j^{-1/2} \frac{(j-r)^{-1/2}}{n-l}.
\end{align*}
By assumption, $k \sim n-l$ and $\log l \sim \log k$. Invoking Lemma~\ref{lem:Integral1} hence yields
\begin{align*}
&\sum_{j=s_\star}^{s^\star} \Prob\left(A_{l,n} \geq j \right) \Prob\left(A_{n-l,n} =k-j \right) \\
& \hspace{3cm}\sim \frac{\beta(1-\beta)\theta^2}{\pi} \frac{\log(k/(|r|+1)}{k}
\end{align*}
as $n \rightarrow \infty$. Using~\eqref{eq:ExceedanceInOneComponent}, we obtain as $n \rightarrow \infty$,
\begin{align*}
\Prob\left(A_{n} \geq k \right) \sim \Prob\left(A_{n-l,n} \geq k \right) + \frac{\beta(1-\beta)\theta^2}{\pi} \frac{\log\left(\frac{k}{|r|+1}\right)}{k},
\end{align*}
where the asymptotic behavior of $\Prob\left(A_{n-l,n} \geq k \right)$ is given by equation~\eqref{eq:ExceedanceInOneComponent}. The result follows by observing that phase transitions occur when $-r \propto \log^2 k$. In words, the threshold is most likely exceeded in the larger component alone, or both components have a significant number of line failures. The latter turns dominant as soon as the difference between the threshold and larger component size becomes small enough.

Next, we prove the second case of the theorem with $\beta=1/2$. That is, the two component sizes are approximately the same, making the analysis more delicate. Effectively, we follow the same strategy as before, but make some modifications as the smaller component is approximately of the same size as the bigger component. Equation~\eqref{eq:ExceedanceInOneComponent} provides the asymptotic behavior of $\Prob\left(A_{n-l,n} \geq k \right)$. Again, let $s_\star$ and $s^\star$ be as above. Using the analysis above shows that
\begin{align*}
&\sum_{j=\max\{1,r\}}^{s^\star} \Prob\left(A_{l,n} \geq j \right) \Prob\left(A_{n-l,n} =k-j \right) \\
&\hspace{.75cm}\sim \frac{\beta(1-\beta)\theta^2}{\pi} \frac{\log(k/(|r|+1)}{k} = \frac{\theta^2}{4\pi} \frac{\log(k/(|r|+1)}{k}
\end{align*}
remains valid in this case, covering the asymptotic behavior of terms II and III. 

Let $q_\star=l-o(l)$ satisfy $l-q_\star= \omega \left( l/ \log(l/(|t|+1))\right)$ and growing. Let $q^\star=k-o(l-k)=l-(|t|+o(|t|))$ be growing such that $k-q^\star = \omega(|t|/\log(l/|t|))$ growing if $-t=l-k>0$ growing, and $q^\star=l-o(l)$ such that $l-q^\star=o(\log(l/(|t|+1)))$ is growing otherwise. We observe that for this choice of $q_\star$, term IV yields
\begin{align*}
&\sum_{j=s^\star}^{q_\star} \Prob\left(A_{l,n} \geq j \right) \Prob\left(A_{n-l,n} =k-j \right)\\
& \leq \Prob\left(A_{l,n} \geq s^\star \right)  \Prob\left(A_{n-l,n} \geq k-q_\star \right)
= O \hspace{-.05cm} \left( \hspace{-.1cm} \frac{1}{\sqrt{s^\star(k-q_\star)}} \hspace{-.05cm}\right)\\
&= o\left( \frac{\sqrt{\log(l/(|r|+1))  \log(l/(|t|+1))}}{k} \right) \\
&= o\left( \frac{\log(k/(|r|+1)+\log(k/(|t|+1)}{k} \right).
\end{align*}

\noindent
It follows from Theorem~\ref{thm:NoSplitting} that uniformly as $n \rightarrow \infty$,
\begin{align*}
\sum_{j=q_\star}^{q^\star} &\Prob\left(A_{l,n} \geq j \right) \Prob\left(A_{n-l,n} =k-j \right) \\
&\hspace{3cm}\sim \frac{\theta^2}{4\pi} \sum_{j=q_\star}^{q^\star} \frac{(l-j)^{1/2}}{l} (k-j)^{-3/2}.
\end{align*}
Applying Lemma~\ref{lem:Integral2} results into $q_\star < q^\star$ for all $n$ large enough, and
\begin{align*}
\sum_{j=q_\star}^{q^\star} \Prob\left(A_{l,n} \geq j \right) \Prob\left(A_{n-l,n} =k-j \right) \sim \frac{\theta^2}{4\pi} \frac{\log(k/(|t|+1)}{k}.
\end{align*}
That is, it describes the asymptotic behavior of all events where almost all lines have failed in the smaller component while the number of failures in the larger component is substantial, yet relatively small. Finally, Lemma~\ref{lem:ConditionBiggerComponentLarge} implies as $n\rightarrow \infty$,
\begin{align*}
\sum_{j=q^\star}^{l} \Prob\left(A_{l,n} \geq j \right) \Prob\left(A_{n-l,n} =k-j \right) 
\sim \frac{\theta}{\sqrt{2\pi}} \frac{\sqrt{-t}}{k} 
\end{align*}
if $-t>0$ growing, and 
\begin{align*}
\sum_{j=q^\star}^{l} \Prob\left(A_{l,n} \geq j \right) \Prob\left(A_{n-l,n} =k-j \right) = o\left( \frac{\log\left(\frac{l}{|t|+1)}\right)}{k}  \right)
\end{align*}
otherwise. In other words, the event that the threshold is exceeded in the smaller component alone contributes to the dominant behavior only if $k$ is significantly smaller than $l$. Combining the above results then concludes the result for $\beta=1/2$.

Finally, we have to show the result for $\alpha=1$. The threshold is close to $n$ itself, and hence both components can only have a few surviving lines after the cascading failure process. Recall~\eqref{eq:SpecialThirdIdentity} and the results of Lemmas~\ref{lem:ConditionBiggerComponentSmall} and~\ref{lem:ConditionBiggerComponentLarge}. Observe that $r>0$ is of order $n$, and hence for any $s_\star=r+o(n-k)$ and $s^\star=l-o(n-k)$ satisfying the conditions in Lemmas~\ref{lem:ConditionBiggerComponentSmall} and~\ref{lem:ConditionBiggerComponentLarge} yield
\begin{align*}
\sum_{j=r}^{s_\star} \Prob\left(A_{l,n} \geq j \right) \Prob\left(A_{n-l,n} =k-j \right) &= o(k^{-1}), \\
\sum_{j=s^\star}^{l} \Prob\left(A_{l,n} \geq j \right) \Prob\left(A_{n-l,n} =k-j \right) &= o\left( \frac{n-k}{k^2}\right).
\end{align*}
To finalize the proof, we hence have to show that for suitable $s_\star$ and $s^\star$ satisfying the conditions above,
\begin{align}
\sum_{j=s_\star}^{s^\star} \Prob\left(A_{l,n} \geq j \right) \Prob\left(A_{n-l,n} =k-j \right) \sim \frac{\theta^2}{2} \left( n-k \right)k^{-2}
\label{eq:ToShowBalanced3}
\end{align}
as $n \rightarrow \infty$. Fix $\epsilon>0$, then for large enough $n$,
\begin{align*}
&\sum_{j=s_\star}^{s^\star} \Prob\left(A_{l,n} \geq j \right) \Prob\left(A_{n-l,n} =k-j \right) \\
&= \sum_{j=s_\star-r}^{s^\star-r} \Prob\left(A_{l,n} \geq j+r \right) \Prob\left(A_{n-l,n} =n-l-j \right) \\
&\leq (1+\epsilon)^2 \int_{x=0}^{n-k} \frac{2\beta \theta}{\sqrt{2\pi}} \frac{\sqrt{l-r-x}}{l} \frac{(1-\beta)\theta}{\sqrt{2\pi}}\frac{1}{\sqrt{x}(n-l)} \, dx \\
&= (1+\epsilon)^2 \frac{\beta(1-\beta)\theta^2}{\pi l (n-l)}\int_{x=0}^{n-k} (n-k-x)^{1/2} x^{-1/2} \, dx \\
&= (1+\epsilon)^2 \frac{\beta(1-\beta)\theta^2}{2} \frac{n-k}{l(n-l)}.
\end{align*}
For the lower bound, note that we can set $s_\star-r=l-s^\star$ without violating the constraints of the lemmas. This is done to simplify the integration term in the lower bound, i.e.
\begin{align*}
&\sum_{j=s_\star}^{s^\star} \Prob\left(A_{l,n} \geq j \right) \Prob\left(A_{n-l,n} =k-j \right) \\
&\geq (1-\epsilon)^2 \frac{\beta(1-\beta)\theta^2}{\pi l (n-l)}\int_{x=s_\star-r+1}^{s^\star-r-1} (n-k-x)^{1/2} x^{-1/2} \, dx \\
&= (1-\epsilon)^2 \frac{\beta(1-\beta)\theta^2(n-k)}{\pi l (n-l)} \\
&\hspace{2cm}\cdot \arctan\left( \frac{n-k-2(s_\star-r)}{2\sqrt{(n-k-(s_\star-r))(s_\star-r)}}\right).
\end{align*}
We note that
\begin{align*}
\lim_{n \rightarrow \infty} \arctan\left( \frac{n-k-2(s_\star-r)}{2\sqrt{(n-k-(s_\star-r))(s_\star-r)}}\right) = \frac{\pi}{2}.
\end{align*}
Letting $\epsilon \downarrow 0$ shows that the bounds coincide and hence~\eqref{eq:ToShowBalanced3} holds.
\end{proof}

\subsection{Proof of Theorem~\ref{thm:NonProportional}}
\begin{proof}[Proof of Theorem~\ref{thm:NonProportional}]
In order to see that~\eqref{eq:NonProportionalNormal} holds for $\alpha<1$, choose $s_\star=o(\min\{k,l\}$ and $s^\star = \min\{k,l\} - o(\min\{k,l\})$. In addition, let $s^\star = \min\{ k,l\} - \omega (\min\{ k,l\}/s_\star)$ if $k =O(l)$. Lemma~\ref{lem:ConditionSmallerComponentSmall} and~\eqref{eq:ExceedanceInOneComponent} yield
\begin{align*}
\Prob\left(A_n \geq k ; A_{l,n} \leq s_\star \right)&\sim \Prob\left(A_{n-l,n} \geq k \right) \\&\sim \frac{2\theta}{\sqrt{2\pi}}  \sqrt{1-\frac{k}{n-l}} k^{-1/2},
\end{align*}
and
\begin{align*}
\Prob&\left(A_n \geq k ; A_{l,n} \in (s_\star,s^\star) \right) \\
&\hspace{1cm}\leq \underbrace{\Prob\left(A_{l,n} \geq s_\star \right)}_{=O(l/n s_\star^{-1/2})}\underbrace{\Prob\left(A_{n-l,n} \geq k-s^\star \right)}_{=O((k-s^\star)^{-1/2})} = o(k^{-1/2}).
\end{align*}
Moreover, if $k<l$,
\begin{align*}
\Prob&\left(A_n \geq k ; A_{l,n} \geq s_\star \right) \leq \Prob\left(A_{l,n} \geq s_\star \right) = O\left( \frac{l}{n} \sqrt{\frac{l-s^\star}{s^\star l}}\right) \\
&\hspace{2cm}=O\left( \frac{\sqrt{l(l-s^\star)}}{n} k^{-1/2} \right) = o(k^{-1/2}),
\end{align*}
and if $k \geq l$,
\begin{align*}
\Prob&\left(A_n \geq k ; A_{l,n} \geq s_\star \right) \leq \Prob\left(A_{l,n} \geq s_\star \right)  = O\left( \frac{l}{n} \sqrt{\frac{l-s^\star}{s^\star l}}\right) \\
&\hspace{3cm}=O\left( \frac{\sqrt{l-s^\star}}{n} \right) = o(k^{-1/2}).
\end{align*}

\noindent
Due to~\eqref{eq:SpecialFirstIdentity}, we can therefore conclude that~\eqref{eq:NonProportionalNormal} holds when $\alpha<1$. 

Next, suppose $\alpha = 1 $. Equation~\eqref{eq:ExceedanceInOneComponent} then translates to
\begin{align*}
\Prob\left(A_{n-l,n} \geq k \right) \left\{ \begin{array}{ll}
\sim \frac{2 \theta}{\sqrt{2\pi}} \frac{\sqrt{-r}}{k} & \textrm{if } -r>0 \textrm{ growing},\\
\sim \chi(r) k^{-1} & \textrm{if } r \textrm{ fixed,} \\
=0 & \textrm{if } r>0,
\end{array} \right.
\end{align*}
where
\begin{align*}
\chi(r)=\sum_{m=0}^{\max\{ -r,\lfloor\theta\rfloor \}} \frac{\theta (m-\theta)^m}{m!} e^{-(m-\theta)} .
\end{align*}
If $-r = \Omega(l)$, we have the bound
\begin{align*}
\Prob\left(A_{n} \geq k ; A_{n-l,n} < k \right) \leq \Prob\left(A_{l,n} \geq 1 \right) \Prob\left(A_{l,n} \geq k-l \right).
\end{align*}
Since $A_{l,n}$ obeys a quasi-binomial distribution~\cite{Dobson2005},
\begin{align*}
\Prob\left( A_{l,n} \geq 1 \right) = 1-\left(1-\frac{\theta}{n}\right)^l = o(1),
\end{align*}
and the second term is bounded by
\begin{align*}
\Prob\left( A_{n-l,n} \geq k-l \right) = O\left( \frac{\sqrt{n-k}}{n-l}\right) =O\left( \frac{\sqrt{-r}}{k}\right).
\end{align*}
Again, due to identity~\eqref{eq:SpecialFirstIdentity}, we observe that~\eqref{eq:NonProportionalNormal} holds in this case as well.

Next, suppose $\alpha =1$ with $|r|=o(l)$. Choose $s_\star$ small enough and $s^\star$ large enough such that the condition in Lemma~\ref{lem:Integral1} and~\ref{lem:DisparateSSmall} are satisfied. Then, uniformly,
\begin{align*}
&\sum_{j=s_\star+1}^{s^\star-1}  \Prob\left(A_{l,n} \geq j \right) \Prob\left(A_{n-l,n} =k-j \right) \\
&\sim \sum_{j=s_\star+1}^{s^\star-1} \frac{\theta^2}{\pi} \frac{l}{k^2} j^{-1/2}  (j-r)^{-1/2} \sim \frac{\theta^2}{\pi} \frac{l}{k^2} \log \left( \frac{l}{|r|+1}\right).
\end{align*}

\noindent
Recalling~\eqref{eq:SpecialFirstIdentity},~\eqref{eq:ExceedanceInOneComponent}, and applying Lemmas~\ref{lem:DisparateSSmall} and~\ref{lem:DisparateSLarge} then yields that as $n \rightarrow \infty$,
\begin{align}
\Prob\left(A_n \geq k \right) \sim \Prob\left(A_{n-l,n} \geq k \right) + \frac{\theta^2}{\pi} \frac{l \log(l/(|r|+1))}{k^2}.
\label{eq:DisparateCriticalExceedance}
\end{align}

It follows immediately that~\eqref{eq:NonProportionalRPositive} holds if $r>0$. Moreover, if $r\leq 0$ and fixed, the exceedance of the threshold in the larger component alone already yields a term of order $k^{-1}$. Since $\log(l/(|r|+1))<\log(n)$ and $k \sim n$, it is necessary that $l = \Omega(n/ \log n)$ for the second term in~\eqref{eq:DisparateCriticalExceedance} to be non-negligible. It is also sufficient since $l=n/\log n$ yields
\begin{align*}
\frac{l \log (l/(|r|+1))}{k^2} \sim \frac{n \log(n/\log n)}{k^2 \log n} \sim \frac{n \log n}{k\,  n \log n} = \frac{1}{k}.
\end{align*}
That is,~\eqref{eq:NonProportionalRFixed} holds as well. When $-r>0$ and growing, the same reasoning as before shows that it is necessary that $l = \Omega(n \sqrt{-r}/ \log n)$ for the second term in~\eqref{eq:DisparateCriticalExceedance} to be non-negligible. Note this implies $\sqrt{-r}=O\left(l/n \log n\right)=o(\log n)$. This condition is also sufficient: when $l= n \sqrt{-r}/ \log n$,
\begin{align*}
\frac{l \log (l/(|r|+1))}{k^2} &= \frac{n \sqrt{-r} \log(n/(\sqrt{-r}\log n))}{k^2 \log n} \\
&\sim \frac{\sqrt{-r}\log n}{k \log n} = \frac{\sqrt{-r}}{k}.
\end{align*}
In conclusion, also the case~\eqref{eq:NonProportionalRVeryNegative} holds.

Next, suppose $\alpha=1$ with $\gamma=\lim_{n \rightarrow \infty} r/l \in (0,1)$. Choose $s_\star=r+o(l)$ such that $s_\star-r$ is growing, and $s^\star=l-o(l)$ such that $l-s^\star$ is growing. Then, 
\begin{align*}
&\Prob\left(A_n \geq k ; A_{n-l,n} \geq k-s_\star\right) \\
&\hspace{2.5cm}\leq \Prob\left( A_{l,n} \geq r \right) \Prob\left( A_{n-l,n} \geq k-s_\star \right) \\
&\hspace{2.5cm} = O\left( \frac{l}{n} r^{-1/2} \frac{(k-s_\star)^{1/2}}{n-l} \right) = o\left(\frac{l}{k^2}\right)
\end{align*}
and
\begin{align*}
\Prob&\left(A_n \geq k ; A_{n-l,n} \leq k-s^\star\right) \\
&\hspace{2cm}\leq \Prob\left( A_{l,n} \geq s^\star \right) \Prob\left( A_{n-l,n} \geq k-l \right) \\
&\hspace{2cm}= O\left( \frac{l}{n} \frac{(l-s^\star)^{1/2}}{l} \frac{\sqrt{n-k}}{n-l} \right) = o\left(\frac{l}{k^2}\right).
\end{align*}

Using~\eqref{eq:ExceedanceInOneComponent}, we obtain that as $n\rightarrow \infty$,
\begin{align*}
\Prob&\left(A_n \geq k ; A_{n-l,n} \in (k-s^\star, k-s_\star) \right) \\
&= \sum_{j=s_\star+1}^{s^\star-1} \Prob\left(A_{n-l,n} = k-j \right) \Prob\left(A_{l,n} \geq j \right)\\
&\sim \sum_{j=s_\star+1}^{s^\star-1} \frac{\theta}{\sqrt{2\pi}} \frac{1}{(n-l)\sqrt{n-l-k+j}} \cdot \frac{2\theta}{\sqrt{2\pi}} \frac{l}{n} \sqrt{\frac{l-j}{l\cdot j}} \\
&\sim \frac{\theta^2}{\pi} \frac{\sqrt{l}}{k^2} \sum_{j=s_\star+1}^{s^\star-1} \sqrt{\frac{l-j}{(j-r)j}}.
\end{align*}
Note that the function within the summation is (strictly) decreasing on $(r,l]$. Hence, an upper bound for the summation term is given by
\begin{align*}
&\sum_{j=s_\star+1}^{s^\star-1} \sqrt{\frac{l-j}{(j-r)j}} \leq \int_{x=s_\star}^l \sqrt{\frac{l-j}{(j-r)j}} \, dx \\
&=\sqrt{l} \int_{y=s_\star/l}^1 \sqrt{\frac{1-y}{(y-r/l)y}} \, dy \sim \sqrt{l} \int_{y=\gamma}^1 \sqrt{\frac{1-y}{(y-\gamma)y}} \, dy,
\end{align*}
and a lower bound is given by
\begin{align*}
&\sum_{j=s_\star+1}^{s^\star-1} \sqrt{\frac{l-j}{(j-r)j}} \geq \int_{x=s_\star+1}^{s^\star} \sqrt{\frac{l-j}{(j-r)j}} \, dx \\
&=\int_{y=(s^\star+1)/l}^{s^\star/l} \sqrt{\frac{l(1-y)}{(y-r/l)y}} \, dy \sim \sqrt{l} \int_{y=\gamma}^1 \sqrt{\frac{1-y}{(y-\gamma)y}} \, dy.
\end{align*}
As the asymptotic behavior of the upper and lower bound coincides, we obtain
\begin{align*}
\Prob&\left(A_n \geq k ; A_{n-l,n} \in (k-s^\star, k-s_\star) \right) \\
&\hspace{3cm}\sim \frac{\theta^2}{\pi} \int_{y=\gamma}^1 \sqrt{\frac{1-y}{(y-\gamma)y}} \, dy \, \frac{\sqrt{l}}{k^2}.
\end{align*}
We observe that $\int_{y=\gamma}^1 \sqrt{(1-y)/((y-\gamma)y)} \, dy$ is a constant, since
\begin{align*}
\int_{y=\gamma}^1 \sqrt{\frac{1-y}{(y-\gamma)y}} \, dy \leq \int_{y=\gamma}^1 \sqrt{\frac{1}{(y-\gamma)\gamma}} \, dy < \infty.
\end{align*}
Recalling~\eqref{eq:SpecialThirdIdentity} yields the result in this case.

Finally, we consider $\alpha=1$ with $r=l-o(l)$, and hence both components can only have a few surviving lines after the cascading failure process. The proof is analogous to the case where $\beta \neq 0$, and is merely adapted below to account for the disparity between the component sizes. Choose $s_\star=r+o(n-k)$ and $s^\star=l-o(n-k)$. Then,
\begin{align*}
\Prob&\left(A_n \geq k ; A_{n-l,n} \geq k-s_\star\right) \\
&\hspace{2cm}\leq \Prob\left( A_{l,n} \geq r \right) \Prob\left( A_{n-l,n} \geq k-s_\star \right) \\
&\hspace{2cm}= O\left( \frac{l}{n} \frac{\sqrt{l-r}}{l} \frac{\sqrt{k-s_\star}}{n-l} \right) = o\left(\frac{{n-k}}{k^2}\right),
\end{align*}
and
\begin{align*}
\Prob&\left(A_n \geq k ; A_{n-l,n} \leq k-s^\star\right) \\
&\hspace{2cm}\leq \Prob\left( A_{l,n} \geq s^\star \right) \Prob\left( A_{n-l,n} \geq k-l \right) \\
&\hspace{2cm}= O\left( \frac{l}{n} \frac{\sqrt{l-s^\star}}{l} \frac{\sqrt{n-k}}{n-l} \right) = o\left(\frac{{n-k}}{k^2}\right).
\end{align*}

Fix $\epsilon>0$. Using~\ref{eq:ExceedanceInOneComponent} yields
\begin{align*}
&\sum_{j=s_\star+1}^{s^\star-1} \Prob\left(A_{l,n} \geq j \right) \Prob\left(A_{n-l,n} =k-j \right) \\
&= \sum_{j=s_\star+1-r}^{s^\star-1-r} \Prob\left(A_{l,n} \geq j+r \right) \Prob\left(A_{n-l,n} =n-l-j \right) \\
&\leq (1+\epsilon)^2 \int_{x=0}^{n-k} \frac{l}{n} \frac{2 \theta}{\sqrt{2\pi}} \frac{\sqrt{l-r-x}}{l} \frac{\theta}{\sqrt{2\pi}}\frac{1}{\sqrt{x}(n-l)} \, dx \\
&= (1+\epsilon)^2 \frac{\theta^2}{\pi n (n-l)}\int_{x=0}^{n-k} (n-k-x)^{1/2} x^{-1/2} \, dx \\
&= (1+\epsilon)^2 \frac{\theta^2}{2} \frac{{n-k}}{n(n-l)}.
\end{align*}
For the lower bound, note that we can set $s_\star-r=l-s^\star$ without violating the assumptions on $s_\star$ and $s^\star$. This is done to simplify the integration term in the lower bound, i.e.
\begin{align*}
&\sum_{j=s_\star+1}^{s^\star-1} \Prob\left(A_{l,n} \geq j \right) \Prob\left(A_{n-l,n} =k-j \right) \\
&\geq (1-\epsilon)^2 \frac{\theta^2}{\pi n (n-l)}\int_{x=s_\star-r+2}^{s^\star-r-2} (n-k-x)^{1/2} x^{-1/2} \, dx \\
&\geq (1-\epsilon)^3 \frac{\theta^2}{\pi n (n-l)} (n-k)\\
&\hspace{2cm}\cdot \arctan\left( \frac{n-k-2(s_\star-r)}{2\sqrt{(n-k-(s_\star-r))(s_\star-r)}}\right).
\end{align*}
We note that
\begin{align*}
\lim_{n \rightarrow \infty} \arctan\left( \frac{n-k-2(s_\star-r)}{2\sqrt{(n-k-(s_\star-r))(s_\star-r)}}\right) = \frac{\pi}{2}.
\end{align*}
Letting $\epsilon \downarrow 0$ shows that the bounds coincide and hence
\begin{align*}
\sum_{j=s_\star+1}^{s^\star-1} \Prob\left(A_{l,n} \geq j \right) \Prob\left(A_{n-l,n} =k-j \right) \sim \frac{\theta^2}{2} \left( n-k \right)k^{-2}.
\end{align*}
Combining the results show that the theorem also holds in this final case.
\end{proof}

\section{Asymptotic behavior of some summation terms}
In our analysis determining the asymptotic behavior often boils down to deriving the asymptotics of some summation terms. In this section we provide two of such results that are used.

\begin{lemma}
Suppose $|r|=o(l)$. Let $s_\star$ be such that $s_\star=o(\log (l/(|r|+1)) )$ growing if $-r>0$ growing or $|r|$ fixed, and $s_\star = r+ o(r)$ with $s^\star-r$ growing if $r>0$ growing. Let $s^\star=o(l)$ be growing such that $s^\star=\omega \left( l/ \log(l/(|r|+1))\right)$. Then, as $l \rightarrow \infty$, $s_\star \lesssim s^\star$ and
\begin{align}
\sum_{j=s_\star}^{s^\star} j^{-1/2}(j-r)^{-1/2} \sim  \log \left(\frac{l}{|r|+1}\right).
\end{align}
\label{lem:Integral1}
\end{lemma}

\noindent
\begin{proof}
First, we have that $s_\star \leq s^\star$ as $l \rightarrow \infty$. That is, if $-r>0$ growing or $|r|$ fixed, 
\begin{align*}
s_\star \lesssim \log \left(\frac{l}{|r|+1}\right) \leq \log(l) \lesssim \frac{l}{\log l} \leq \frac{l}{\log \left(\frac{l}{|r|+1}\right)} \lesssim s^\star,
\end{align*}
and if $r>0$ growing,
\begin{align*}
s_\star \sim r = l \, \frac{r}{l}  \lesssim \frac{l}{\log (l/r)}  \lesssim s^\star.
\end{align*}

\noindent
Next, observe that the expression in the summation is a decreasing function, and therefore
\begin{align*}
\sum_{j=s_\star}^{s^\star} j^{-1/2}(j-r)^{-1/2} &\leq \int_{j=s_\star-1}^{s^\star} j^{-1/2}(j-r)^{-1/2} \\
&= 2 \log \left( \frac{\sqrt{s^\star}+\sqrt{s^\star-r}}{\sqrt{s_\star-1}+\sqrt{s_\star-r-1}}\right),
\end{align*}
and
\begin{align*}
\sum_{j=s_\star}^{s^\star} j^{-1/2}(j-r)^{-1/2} &\geq \int_{j=s_\star}^{s^\star} j^{-1/2}(j-r)^{-1/2} \\
&= 2 \log \left( \frac{\sqrt{s^\star}+\sqrt{s^\star-r}}{\sqrt{s_\star}+\sqrt{s_\star-r}}\right).
\end{align*}
It is apparent that the asymptotic behavior of the upper bound and lower bound is the same. It remains to derive this behavior in terms of $l$ and $r$. 

For an asymptotic upper bound, we observe that $\sqrt{s^\star}+\sqrt{s^\star-r} \lesssim \sqrt{l}$ and $\sqrt{s^\star}+\sqrt{s^\star-r} \gtrsim \sqrt{|r|+1}$ due to our choice of $s^\star$. Therefore
\begin{align*}
\sum_{j=s_\star}^{s^\star} j^{-1/2}(j-r)^{-1/2} \lesssim 2 \log\left( \sqrt{\frac{l}{|r|+1}} \right)  = \log \left(\frac{l}{|r|}\right).
\end{align*}
For a lower bound, recall that $|r| \lesssim l /\log(l/(|r|+1))$ and thus $\sqrt{s^\star}+\sqrt{s^\star-r} \gtrsim 2 \sqrt{l /\log(l/(|r|+1))}$. Since $\log\log x =o(\log x)$, we derive that as $l \rightarrow \infty$,
\begin{align*}
&\sum_{j=s_\star}^{s^\star} j^{-1/2}(j-r)^{-1/2} \\
&\gtrsim \log \left( \frac{4 \cdot {l /\log(l/(|r|+1))} }{{\max\{|r|,\log(l/(|r|+1)) \}}}\right) \sim \log \left(\frac{l}{|r|+1}\right).
\end{align*}
\end{proof}

\begin{lemma}
Suppose $|r|=o(l)$. Let $s_\star$ be such that $s_\star=o(\log(l/(|r|+1)))$ growing if $-r>0$ growing or $|r|$ fixed, and if $r>0$ growing let $s_\star  = r+ o(r)$ be such that $s_\star-r = \omega(r/\log(l/r))$ growing. Let $s^\star=o(l)$ be growing such that $s^\star=\omega \left( l/ \log(l/(|r|+1))\right)$. Then, there exists a $s_\star$ satisfying the assumptions, as $l \rightarrow \infty$, $s_\star \lesssim s^\star$ and
\begin{align}
\sum_{j=s_\star}^{s^\star} j^{-1/2}(j-r)^{-3/2} \sim  \log \left(\frac{l}{|r|+1}\right).
\label{eq:Integral2}
\end{align}
\label{lem:Integral2}
\end{lemma}

\noindent
\begin{proof}
It is not immediate that if $r>0$ growing, there exists a $s_\star$ that satisfies both $s_\star  = r+ o(r)$ and $s_\star-r = \omega(r/\log(l/r))$. Yet, we observe that $\log(l/r) \rightarrow \infty$ as $l \rightarrow \infty$ and hence $r/\log(l/r) = o(r)$. Therefore, there exists a $s_\star$ that satisfies the stated conditions.

The claim that $s_\star \lesssim s^\star$ as $l \rightarrow \infty$ is already proven in Lemma~\ref{lem:Integral1}.

Finally we have to show~\eqref{eq:Integral2}. Note that the expression in the summation is a decreasing function, and therefore
\begin{align*}
\sum_{j=s_\star}^{s^\star}& j^{1/2}(j-r)^{-3/2} \leq \int_{j=s_\star-1}^{s^\star} j^{1/2}(j-r)^{-3/2} \\
&= 2\sqrt{\frac{s_\star-1}{s_\star-1-r}} -  2\sqrt{\frac{s^\star}{s_\star-r}} \\
&\hspace{2.5cm}+ 2 \log \left( \frac{\sqrt{s^\star}+\sqrt{s^\star-r}}{\sqrt{s_\star-1}+\sqrt{s_\star-1-r}}\right).
\end{align*}
Similarly,
\begin{align*}
\sum_{j=s_\star}^{s^\star} & j^{-1/2}(j-r)^{-1/2} \geq \int_{j=s_\star}^{s^\star} j^{-1/2}(j-r)^{-1/2} \\
&= 2\sqrt{\frac{s_\star}{s_\star-r}} -  2\sqrt{\frac{s^\star}{s_\star-r}} \\
& \hspace{3cm}+ 2 \log \left( \frac{\sqrt{s^\star}+\sqrt{s^\star-r}}{\sqrt{s_\star}+\sqrt{s_\star-r}}\right).
\end{align*}

\noindent
It is apparent that the bounds asymptotically coincide, and it remains to express the asymptotics in terms of $l$ and $r$. First, as we have seen in the proof of Lemma~\ref{lem:Integral1}, $r = O(s^\star)$, and hence
\begin{align*}
2\sqrt{\frac{s^\star}{s^\star-r}} = O(1) = o\left( \log\left(\frac{l}{|r|+1} \right)\right).
\end{align*}
Next, if $r \leq 0$ or $|r|$ fixed, then clearly,
\begin{align*}
2\sqrt{\frac{s_\star}{s_\star-r}} = O(1) = o\left( \log\left(\frac{l}{|r|+1} \right)\right).
\end{align*}
If $r>0$, then
\begin{align*}
2\sqrt{\frac{s_\star}{s_\star-r}} \sim 2 \frac{r}{s_\star-r} = o\left( \log\left(\frac{l}{r} \right)\right).
\end{align*}
Finally, it follows from the proof of Lemma~\ref{lem:Integral1} that
\begin{align*}
2 \log \left( \frac{\sqrt{s^\star}+\sqrt{s^\star-r}}{\sqrt{s_\star}+\sqrt{s_\star-r}}\right) \sim \log\left(\frac{l}{|r|+1} \right),
\end{align*}
as $l \rightarrow \infty$. Adding the above expressions yields the result.
\end{proof}

\section{Proofs of Lemmas~\ref{lem:ConditionSmallerComponentSmall}-\ref{lem:DisparateSLarge}}
\label{app:FewManyProofs}
\begin{proof}[Proof of Lemma~\ref{lem:ConditionSmallerComponentSmall}]
Note that if $1-\beta < \alpha <1$, we must be in the balanced case. Therefore, for large enough $n$, $s_\star \leq k-(n-l)$, which proves the second assertion. 

Next, suppose $\alpha < 1-\beta$ (this can be either the disparate or the balanced case). We then have to prove that the joint event that the threshold is exceeded and the smaller component has few line failures is dominated by the event that $k$ is exceeded in the larger component. Note that $\{A_n \geq k ; A_{l,n} \leq s_\star \}$ implies that at least $\{A_{n-l,n} \geq k-s_\star\}$. Moreover, $\{A_{n-l,n} \geq k\}$ implies $\{A_{n} \geq k\}$. Then,
\begin{align*}
&\underbrace{\left(1- \Prob(A_{l,n} > s_\star) \right)}_{= 1-o(1)} \underbrace{\Prob\left(A_{n-l,n} \geq k \right)}_{\sim \frac{2(1-\beta) \theta}{\sqrt{2\pi}}} \\
&\hspace{1cm}\leq \Prob\left(A_n \geq k ; A_{l,n} \leq s_\star \right) \leq \underbrace{\Prob\left(A_{n-l,n} \geq k-s^\star \right)}_{\sim \frac{2(1-\beta) \theta}{\sqrt{2\pi}}}.
\end{align*}
Therefore,
\begin{align*}
\Prob\left(A_n \geq k ; A_{l,n} \leq s_\star \right) &\sim \frac{2(1-\beta) \theta}{\sqrt{2\pi}} \sqrt{\frac{n-l-k}{k(n-l)}}\\
& \sim \Prob\left(A_{n-l,n} \geq k \right).
\end{align*}
\end{proof}

\noindent
\begin{proof}[Proof of Lemma~\ref{lem:ConditionSmallerComponentLarge}]
First, suppose $\alpha < \beta$. Then we must be in the balanced case, and $k < l$ and $s^\star=k-o(k)$ for $n$ large enough. Basically, we want to show in this case that it is most likely that $A_{l,n}$ already exceeds $k$ given that it exceeds $s^\star$. Note
\begin{align*}
\Prob\left(A_{l,n} \geq k \right) \leq \Prob\left(A_n \geq k ; A_{l,n} \geq s^\star \right) \leq \Prob\left(A_{l,n} \geq s^\star \right).
\end{align*}
Equation~\eqref{eq:ExceedanceInOneComponent} where $l$ is balanced and $\alpha< \beta$ yields that $\Prob\left(A_{l,n} \geq s^\star \right) \sim \Prob\left(A_{l,n} \geq k \right)$ as $n \rightarrow \infty$. This coincides with the lower bound, and hence 
\begin{align*}
\Prob\left(A_n \geq k ; A_{l,n} \geq s^\star \right) \sim \left( A_{l,n} \geq k \right).
\end{align*}

If $\beta \leq \alpha <1-\beta$, we can have both the disparate and the balanced case. When the component sizes are disparate,
\begin{align*}
\Prob&\left(A_n \geq k ; A_{l,n} \geq s^\star \right) \leq  \Prob\left(A_{l,n} \geq s^\star \right) \\
&\hspace{2.5cm}= O\left(\frac{l}{n}\sqrt{\frac{l-s^\star}{l}} {s^{\star}}^{-1/2} \right) =  o(k^{-1/2}).
\end{align*}
When the component sizes are balanced, note that the condition $\beta \leq \alpha <1-\beta$ implies that $(l-s^\star)/l=o(1)$, and hence
\begin{align*}
\Prob&\left(A_n \geq k ; A_{l,n} \geq s^\star \right) \leq  \Prob\left(A_{l,n} \geq s^\star \right) \\
&\hspace{2.5cm}= O\left(\sqrt{\frac{l-s^\star}{l}} {s^{\star}}^{-1/2} \right) =  o(k^{-1/2}).
\end{align*}

Finally, if $1-\beta < \alpha <1$, we have a balanced case and $k>l$ for $n$ large enough. Then,
\begin{align*}
\Prob&\left(A_n \geq k ; A_{l,n} \geq s_\star \right) \\
&\hspace{1cm}\leq  \underbrace{\Prob\left(A_{l,n} \geq s_\star \right)}_{=o(k^{-1/2})} \underbrace{\Prob\left(A_{n-l,n} \geq  k-l \right)}_{=O\left(k^{-1/2} \right)} = o(k^{-1}).
\end{align*}
\end{proof}

\noindent
\begin{proof}[Proof of Lemma~\ref{lem:ConditionBiggerComponentSmall}]
For the first claim, note that $s_\star-r>0$ is growing, and
\begin{align*}
\Prob&\left(A_n \geq k, A_{n-l,n} \in [k-s_\star,\min\{k-1,n-l\}] \right) \\
&\leq s_\star \sup_{i \in [0,s_\star-r]} \Prob\left(A_{n-l,n} =n-l-i \right) = O\left(\frac{s_\star}{n-l} \right).
\end{align*}
Next, in the second case,
\begin{align*}
\Prob&\left(A_n \geq k, A_{n-l,n} \geq k-s_\star \right) \\
&\leq \underbrace{\Prob\left(A_{l,n} \geq r \right)}_{=O(r^{-1/2})} \underbrace{\Prob\left(A_{n-l,n} \geq k-s_\star \right)}_{=O\left( \sqrt{s_\star-r}/(n-l)\right)} = o\left(\frac{\sqrt{s_\star-r}}{n-l} \right).
\end{align*}
For the final case, observe that $r=k-(n-l)=l-(n-k)=l-o(l)$, and hence
\begin{align*}
\Prob&\left(A_n \geq k, A_{n-l,n} \geq k-s_\star \right) \\
&\leq \underbrace{\Prob\left(A_{l,n} \geq r \right)}_{=O(\sqrt{n-k}/l)}\underbrace{\Prob\left(A_{n-l,n} \geq k-s^\star \right)}_{=o(\sqrt{n-k}/(n-l))} = o\left(\frac{n-k}{k^2} \right).
\end{align*}
\end{proof}

\noindent
\begin{proof}[Proof of Lemma~\ref{lem:ConditionBiggerComponentLarge}]
For~\eqref{eq:ConditionBiggerComponentLarge1}, note that $k-l$ is of order $n$ and hence,
\begin{align*}
&\Prob\left(A_n \geq k, A_{n-l,n} \in [k-l,k-s^\star] \right) \\
&\leq \Prob\left(A_{l,n} \geq s^\star \right) \Prob\left(A_{n-l,n} \geq k-l \right) = O\left({s^\star}^{-1/2}k^{-1/2} \right).
\end{align*}
For~\eqref{eq:ConditionBiggerComponentLarge2}, observe $l \sim k$ as $n \rightarrow \infty$ and hence,
\begin{align*}
\Prob&\left(A_n \geq k, A_{n-l,n} \in [k-l,k-s^\star] \right)\\
&\hspace{2cm}\leq \Prob\left(A_{l,n} \geq s^\star \right) = O\left(\frac{\sqrt{l-s^\star}}{k}  \right).
\end{align*}
For~\eqref{eq:ConditionBiggerComponentLarge3}, we thus want to show that it is most likely that if $A_{n-l,n}$ is at most $k-s^\star=o(l-k)$, the threshold is exceeded in the smaller component itself. As $n\rightarrow\infty$,
\begin{align*}
\Prob&\left(A_n \geq k, A_{n-l,n} \in [k-l,k-s^\star] \right) \leq \Prob\left(A_{l,n} \geq s^\star \right) \\
&\sim \frac{2\theta\cdot 1/2}{\sqrt{2\pi}} \sqrt{\frac{l-s^\star}{l s^\star}} \sim \frac{\theta}{\sqrt{2\pi}} \frac{\sqrt{l-k}}{k}.
\end{align*}
For the lower bound,
\begin{align*}
\Prob&\left(A_n \geq k, A_{n-l,n} \in [k-l,k-s^\star] \right) \\
&\geq \Prob\left(A_{l,n} \geq k \right) \Prob\left(A_{n-l,n} \leq k-s^\star \right) \sim \frac{\theta}{\sqrt{2\pi}} \frac{\sqrt{l-k}}{k}.
\end{align*}
Since the lower and upper bounds coincide, we observe that~\eqref{eq:ConditionBiggerComponentLarge3} holds. Finally, if $\alpha=1$ and $\beta \in (0,1/2]$, then
\begin{align*}
\Prob&\left(A_n \geq k, A_{n-l,n} \in [k-l,k-s^\star] \right) \\
&\hspace{1cm}\leq  \Prob\left(A_{l,n} \geq s^\star \right) \Prob\left(A_{n-l,n} \geq n-l-(n-k) \right)\\
&\hspace{1cm}= O\left( \frac{\sqrt{l-s^\star}}{l} \frac{\sqrt{n-k}}{n-l} \right) = o\left( \frac{n-k}{k^2}\right).
\end{align*}
\end{proof}

\noindent
\begin{proof}[Proof of Lemma~\ref{lem:DisparateSSmall}]
When $r<0$ and growing, we have (for $n$ large enough)
\begin{align*}
\Prob&\left(A_n \geq k ; A_{n-l,n} \in [k-s_\star,k-1] \right) \\
&\hspace{1.5cm}= \sum_{j=1}^{s_\star} \Prob\left( A_{l,n} \geq j\right) \Prob\left( A_{n-l,n} = k-j \right) \\
&\hspace{1.5cm}\leq  \Prob\left( A_{l,n} \geq 1\right) s_\star \sup_{j \in [1,s_\star]} \Prob\left( A_{n-l,n} = k-j  \right).
\end{align*}

\noindent
Due to our choice of $s_\star$, we obtain the inequality,
\begin{align*}
\Prob&\left(A_n \geq k ; A_{n-l,n} \in [k-s_\star,k-1] \right) \\
&\hspace{3cm}= O\left( \frac{l}{k} s_\star \frac{1}{\sqrt{-r} k}\right) = o \left(\frac{l \log l}{k^2} \right).
\end{align*}

When $r$ is fixed,
\begin{align*}
\Prob&\left(A_n \geq k ; A_{n-l,n} \in [k-s_\star,k-1] \right) \\
&\hspace{1.5cm}\leq \Prob\left( A_{l,n} \geq 1\right) s_\star \sup_{j \in [1,s_\star]} \Prob\left( A_{n-l,n} = k-j  \right) \\
&\hspace{1.5cm}= O\left( \frac{l}{k} s_\star \frac{1}{ k}\right) = o \left(\frac{l \log l}{k^2} \right).
\end{align*}

When $r>0$ and growing,
\begin{align*}
\Prob&\left(A_n \geq k ; A_{n-l,n} \in [k-s_\star,k-1] \right) \\
&\hspace{1.5cm}\leq \Prob\left( A_{l,n} \geq r\right) \Prob\left( A_{n-l,n} \geq n-l-v  \right) \\
&\hspace{1.5cm}= O\left( \frac{l}{k} r^{-1/2} \frac{\sqrt{v}}{k}\right) = o\left(\frac{l \log(l/r)}{k^2}  \right).
\end{align*}
\end{proof}

\noindent
\begin{proof}[Proof of Lemma~\ref{lem:DisparateSLarge}]
Observe that
\begin{align*}
\Prob&\left(A_n \geq k; A_{n-l,n} \in [k-l,k-s^\star] \right) \\
&\leq \underbrace{\Prob\left( A_{l,n} \geq s^\star \right)}_{=O\left(  \frac{l}{k} {s^\star}^{-1/2} \frac{\sqrt{l}}{k}\right)} \hspace{.5cm}\cdot  \underbrace{\Prob\left( A_{n-l,n} \geq k-l \right)}_{=o \left( \mathbbm{1}_{\{r \leq 0\}}\left( \frac{l \log l}{k^2} \right)+\mathbbm{1}_{\{r>0\}} \frac{l \log(l/r)}{k^2} \right)}.
\end{align*}
\end{proof}

\bibliographystyle{plain}

\begin{thebibliography}{10}

\bibitem{Albert2002}
R.~Albert and A.~Barab\'asi.
\newblock Statistical mechanics of complex networks.
\newblock {\em Rev. Mod. Phys.}, 74:47--97, 2002.

\bibitem{Albert2000}
R.~Albert, H.~Jeong, and A.L. Barabasi.
\newblock Error and attack tolerance of complex networks.
\newblock {\em Nature}, 406(6794):378--382, 2000.

\bibitem{Albert2004}
Rf~Albert, I.~Albert, and G.~L. Nakarado.
\newblock Structural vulnerability of the north american power grid.
\newblock {\em Phys. Rev. E}, 69:025103, 2004.

\bibitem{Bienstock2015}
D.~Bienstock.
\newblock {\em Electrical Transmission System Cascades and Vulnerability}.
\newblock SIAM, 2015.

\bibitem{BlanchetUnpublished}
J.~Blanchet, J.~Li, and M.~Nakayama.
\newblock Efficient {M}onte {C}arlo methods for estimating failure
  probabilities of a distribution network with random demands.
\newblock arXiv:1706.05602.

\bibitem{Carreras2000}
B.~A. Carreras, D.~E. Newman, I.~Dobson, and A.~B. Poole.
\newblock Initial evidence for self-organized criticality in electric power
  system blackouts.
\newblock In {\em P. 33rd HICSS}, 2000.

\bibitem{Carreras2004}
B.~A. Carreras, D.~E. Newman, I.~Dobson, and A.~B. Poole.
\newblock Evidence for self-organized criticality in a time series of electric
  power system blackouts.
\newblock {\em IEEE T. Circuits Syst.}, 51(9):1733--1740, 2004.

\bibitem{Cohen2000}
R.~Cohen, K.~Erez, D.~Ben-Avraham, and S.~Havlin.
\newblock Resilience of the internet to random breakdowns.
\newblock {\em Phys. Rev. Lett.}, 85:4626--4628, 2000.

\bibitem{Cohen2001}
R.~Cohen, K.~Erez, D.~Ben-Avraham, and S.~Havlin.
\newblock Breakdown of the internet under intentional attack.
\newblock {\em Phys. Rev. Lett.}, 86:3682--3685, 2001.

\bibitem{Crucitti2004}
P.~Crucitti, V.~Latora, and M.~Marchiori.
\newblock Model for cascading failures in complex networks.
\newblock {\em Phys. Rev. E}, 69:045104, 2004.

\bibitem{Dobson2004}
I.~Dobson, B.~A. Carreras, and D.~E. Newman.
\newblock A branching process approximation to cascading load-dependent system
  failure.
\newblock In {\em P. 37th HICSS}, 2004.

\bibitem{Dobson2005}
I.~Dobson, B.~A. Carreras, and D.~E. Newman.
\newblock A loading-dependent model of probabilistic cascading failure.
\newblock {\em Probab. Eng. Inform. Sc.}, 19:15--32, 2005.

\bibitem{Guo2016}
J.~Guo, F.~Liu, J.~Wang, J.~Lin, and S.~Mei.
\newblock Towards high-efficiency cascading outage simulation and analysis in
  power systems: a sequential importance sampling approach.
\newblock arXiv:1611.00582v1.

\bibitem{Heide2008}
D.~Heide, M.~Sch\"afer, and M.~Greiner.
\newblock Robustness of networks against fluctuation-induced cascading
  failures.
\newblock {\em Phys. Rev. E}, 77:056103, 2008.

\bibitem{Hines2009}
P.~Hines, J.~Apt, and S.~Talukdar.
\newblock Large blackouts in {North America}: Historical trends and policy
  implications.
\newblock {\em Energy Policy}, 37(12):5249--5259, 2009.

\bibitem{Kim2013}
J.~Kim, J.~A. Bucklew, and I.~Dobson.
\newblock Splitting method for speedy simulation of cascading blackouts.
\newblock {\em IEEE T. Power Syst.}, 28(3):3010--3017, 2013.

\bibitem{Kinney2005}
R.~Kinney, P.~Crucitti, R.~Albert, and V.~Latora.
\newblock Modeling cascading failures in the north american power grid.
\newblock {\em The European Physical Journal B - Condensed Matter and Complex
  Systems}, 46(1):101--107, 2005.

\bibitem{Mirzasoleiman2011}
B.~Mirzasoleiman, M.~Babaei, M.~Jalili, and M.~A. Safari.
\newblock Cascaded failures in weighted networks.
\newblock {\em Phys. Rev. E}, 84:046114, 2011.

\bibitem{Motter2004}
A.~E. Motter.
\newblock Cascade control and defense in complex networks.
\newblock {\em Phys. Rev. Lett.}, 93:098701, 2004.

\bibitem{Motter2002}
A.~E. Motter and Y-C. Lai.
\newblock Cascade-based attacks on complex networks.
\newblock {\em Phys. Rev. E}, 66(065102), 2002.

\bibitem{Newman2011}
D.~E. Newman, B.~A. Carreras, V.~E. Lynch, and I.~Dobson.
\newblock Exploring complex systems aspects of blackout risk and mitigation.
\newblock {\em IEEE T. Reliab.}, 60(1):134--143, 2011.

\bibitem{Pradhan2010}
S.~Pradhan, A.~Hansen, and B.K. Chakrabarti.
\newblock Failure processes in elastic fiber bundles.
\newblock {\em Rev. Mod. Phys.}, 82(1):499--555, 2010.

\bibitem{Rohden2016}
M.~Rohden, D.~Jung, S.~Tamrakar, and S.~Kettemann.
\newblock Cascading failures in ac electricity grids.
\newblock {\em Phys. Rev. E}, 94:032209, 2016.

\bibitem{Shortle2013}
J.~F. Shortle.
\newblock Efficient simulation of blackout probabilities using splitting.
\newblock {\em Elect. Power Energy Syst.}, 44(1):743--751, 2013.

\bibitem{Sloothaak2017}
F.~Sloothaak, S.~C. Borst, and B.~Zwart.
\newblock Impact of network splitting on cascading failure blackouts.
\newblock In {\em IEEE PESGM}, 2017.

\bibitem{Sloothaak2016}
F.~Sloothaak, S.~C. Borst, and B.~Zwart.
\newblock The robustness of power-law behavior in cascading failure models.
\newblock To appear in \emph{Stoch. Models}, arXiv:1604.03677, 2017.

\bibitem{Soltan2014}
S.~Soltan, D.~Mazauric, and G.~Zussman.
\newblock Cascading failures in power grids: Analysis and algorithms.
\newblock In {\em ACM e-Energy: P. 5th IC Fut. Energy Syst.}, pages 195--206,
  2014.

\bibitem{Talukdar2003}
S.~Talukdar, J.~Apt, M.~Ilic, L.~B. Lave, and M.~G. Morgan.
\newblock Cascading failures: Survival versus prevention.
\newblock {\em The Elec. J.}, 16(9):25--31, 2003.

\bibitem{Wang2015}
S.~P. {Wang}, A.~Chen, C.~W. Liu, C.~H. Chen, and J.~F. Shortle.
\newblock Efficient splitting simulation for blackout analysis.
\newblock {\em IEEE T. Power Syst.}, 30(4):1775--1783, 2015.

\bibitem{Watts2002}
D.~J. Watts.
\newblock A simple model of global cascades on random networks.
\newblock {\em P. Nat. Acad. Sci. USA}, 99(9):5766--5771, 2002.

\bibitem{Witthaut2015}
D.~Witthaut and M.~Timme.
\newblock Nonlocal effects and countermeasures in cascading failures.
\newblock {\em Phys. Rev. E}, 92:032809, 2015.

\end{thebibliography}

\end{document}